\documentclass[runningheads]{llncs}
\usepackage[utf8]{inputenc}   
\usepackage[T1]{fontenc}
\usepackage[inline]{enumitem}

\usepackage{centernot}
\usepackage[normalem]{ulem}
\usepackage{xspace}
\usepackage[textsize=small]{todonotes}

\usepackage{etoolbox}
\AtEndEnvironment{proof}{\qed}

\usepackage{amsmath,amsfonts,amssymb,bbm}
\usepackage{hyperref}      
\usepackage[capitalise,noabbrev]{cleveref}

\newcommand{\nop}[1]{}

\DeclareMathOperator{\ActiveOp}{active}
\DeclareMathOperator{\CandOp}{cand}
\DeclareMathOperator{\CanonicalOp}{canonical}
\DeclareMathOperator{\CellsOp}{cells}
\DeclareMathOperator{\IndexOp}{index}
\DeclareMathOperator{\IntervalOp}{interval}

\newcommand{\Prof}{\pi}

\newcommand{\FacB}[2]{f_{#2}(#1)}
\newcommand{\Served}[2]{S_{#2}(#1)}
\newcommand{\Frac}[3]{w_{#2}(#1,#3)}
\newcommand{\Weight}[2]{w_{#2}(#1)}
\newcommand{\Gap}[1]{\Delta(#1)}

\newtheorem{fact}{Fact}

\newcommand{\punt}[1]{}

\newcommand{\Indicator}[1]{\mathbbm{1}_{#1}}

\newcommand{\Reals}{\mathbb{R}}
\newcommand{\Point}{p}
\newcommand{\SolnOne}{\sigma}
\newcommand{\Soln}{\tau}
\newcommand{\Mech}{M}
\newcommand{\MechOpt}{\Mech_1}
\newcommand{\MechLip}{\Mech_2}
\newcommand{\MechOptTwo}{\Mech_3}
\newcommand{\MechLipTwo}{\Mech_4}
\newcommand{\Dist}[2]{d(#1,#2)} 
\newcommand{\CostSoln}[2]{C(#1,#2)} 
\newcommand{\ServedSoln}[3]{S_{#2}(#1,#3)} 

\newcommand{\Canonical}[1]{\CanonicalOp(#1)}
\newcommand{\Index}[1]{\IndexOp(#1)}

\newcommand{\AgentSet}{[n]}
\newcommand{\Cost}[1]{C(#1)}
\newcommand{\CostPrime}[1]{C'(#1)}
\newcommand{\Cand}[2]{\CandOp(#1,#2)}
\newcommand{\Fac}[2]{f_{#2}(#1)}
\newcommand{\FacPrime}[2]{f'_{#2}(#1)}
\newcommand{\FracPrime}[3]{w'_{#2,#3}(#1)}
\newcommand{\WeightPrime}[2]{w'_{#2}(#1)}
\newcommand{\GapPrime}[1]{\Delta'(#1)}

\newcommand{\CostDensity}[2]{\psi_{#2}(#1)}
\newcommand{\CostDensityPrime}[2]{\psi'_{#2}(#1)}
\newcommand{\BackoffHelper}[2]{\xi_{#2}(#1)}

\newcommand{\Backoff}[2]{\varphi_{#2}(#1)}

\newcommand{\Active}[1]{\ActiveOp(#1)}
\newcommand{\Interval}[2]{\IntervalOp(#1,#2)}
\newcommand{\IntervalA}{I}
\newcommand{\Hood}[1]{\Gamma(#1)}
\newcommand{\HelperFn}[1]{h_{#1}}
\newcommand{\Helper}[2]{\HelperFn{#2}(#1)}
\newcommand{\HelperPrime}[2]{\HelperFn{#2}'(#1)}
\newcommand{\Interior}[1]{\mathcal{I}(#1)}
\newcommand{\Knots}{K}
\newcommand{\Cells}[1]{\CellsOp(#1)}
\newcommand{\Cell}{X}
\newcommand{\Mean}[1]{\mu(#1)}
\newcommand{\MeanPrime}[1]{\mu'(#1)}

\newcommand{\eps}{\varepsilon}
\DeclareMathOperator*{\SubstOp}{subst}
\newcommand{\Subst}[3]{\SubstOp(#1,#2,#3)}
\newcommand{\ProfLie}{\Prof^*}
\newcommand{\SolnLie}{\Soln^*}
\newcommand{\SolnOneLie}{\SolnOne^*}
\newcommand{\Set}[1]{\left\lbrace #1\right\rbrace}

\newcommand{\FactorLetter}{\alpha}
\newcommand{\Factor}[1]{\FactorLetter(#1)}
\newcommand{\LipConstant}{\kappa}

\spnewtheorem{mechanism}{Mechanism}{\bfseries}{\itshape}

\usepackage{color}

\makeatletter
\@addtoreset{lemma}{section}          
\makeatother

\usepackage{graphicx}

\usepackage{color}

\begin{document}
\title{Constant-Approximate and Constant-Strategyproof Two-Facility Location}
%
%
%
\author{Elijah Journey Fullerton\inst{1} \and
Zeyuan Hu\inst{2} 
\and 
C. Gregory Plaxton \inst{2}}
\authorrunning{Fullerton, Hu, and Plaxton}
%
\institute{Princeton University
\email{(ef0952@princeton.edu)}\\ \and
University of Texas at Austin\\
\email{(zeyuan.zack.hu@gmail.com, plaxton@cs.utexas.edu)}}
\maketitle              
\begin{abstract}
We study deterministic mechanisms for the two-facility location
problem.  Given the reported locations of $n$ agents on the real
line, such a mechanism specifies where to build the two
facilities.  The single-facility variant of this problem admits a
simple strategyproof mechanism that minimizes social cost. For two
facilities, however, it is known that any strategyproof mechanism is
$\Omega(n)$-approximate. We seek to circumvent this strong lower bound
by relaxing the problem requirements.  Following other work in the
facility location literature, we consider a relaxed form of
strategyproofness in which no agent can lie and improve their outcome
by more than a constant factor.  Because the aforementioned
$\Omega(n)$ lower bound generalizes easily to constant-strategyproof
mechanisms, we introduce a second relaxation: Allowing the facilities
(but not the agents) to be located in the plane.  Our first main
result is a natural mechanism for this relaxation that is constant-approximate and constant-strategyproof. A characteristic of this mechanism is that a
small change in the input profile can produce a large change in the
solution. Motivated by this observation, and also by results in the
facility reallocation literature, our second main result is a
constant-approximate, constant-strategyproof, and Lipschitz continuous
mechanism.

\keywords{Facility Location Problem \and Mechanism Design}
\end{abstract}
\setcounter{footnote}{0}
\section{Introduction}
\label{sec:introduction}

Facility location is a canonical problem in algorithmic game theory
and mechanism design. The typical problem formulation is that given a
set of agent locations, a facility location mechanism determines a set
of locations for the facilities. Each agent then has a service cost
given by the distance between their location and their nearest
facility. The social cost is defined as the sum of the agent service
costs. It is desirable for a mechanism to be efficient, that is, to
minimize the social cost (under the assumption of truthful reporting).
To encourage truthful reporting, it is
desirable for a mechanism to be strategyproof, which means that an
agent cannot decrease their service cost by misreporting.

In this
paper, we focus on deterministic mechanisms for the well-studied
setting in which the agents are located on the real line.  While our
focus is on the design of mechanisms for locating two facilities, it
is useful to briefly comment on the case of a single facility.  As
discussed by Procaccia and Tennenholtz~\cite{Procaccia2013}, it is
easy to design a single-facility location mechanism that is efficient
and strategyproof; the idea is to place the facility at the median
agent location when the number of agents is odd and at either the left
or right middle agent location when the number of agents is
even.\footnote{More general results of Moulin
\cite{Moulin1980} establish that for agent locations on the real
number line and single-peaked preferences, these are the only
mechanisms for single-facility location that are efficient and
strategyproof.}

When the number of facilities is expanded to two, strategyproofness
starts to restrict the efficiency that a mechanism can provide. Lu et
al.~\cite{Lu2010} show that any strategyproof mechanism has
$\Omega(n)$ approximation ratio of the social cost.\footnote{The simple mechanism
that locates the facilities at the leftmost and the rightmost agents achieves $O(n)$ approximation of the social cost,
which indicates that the $\Omega(n)$ bound is asymptotically tight \cite{Procaccia2013}. In fact, Fotakis and
Tzamos~\cite{Fotakis2014} show that the foregoing mechanism achieves
the best possible approximation ratio for the problem and is the only
deterministic anonymous strategyproof mechanism with approximation
ratio bounded by a function of $n$.}
In this paper, we seek to circumvent the $\Omega(n)$ lower bound by
relaxing the problem requirements. The first relaxation is to target an approximate form of strategyproofness. In particular, we seek a mechanism that is  constant-strategyproof~\cite{Oomine2017} instead of
strategyproof. A constant-strategyproof mechanism ensures that no agent
can reduce their service distance by more than a multiplicative constant
factor by misreporting. It is straightforward to generalize the $\Omega(n)$ result of Lu et al.~\cite{Lu2010} to show that the same $\Omega(n)$ lower bound holds for constant-strategyproof mechanisms. Accordingly, we introduce a second
relaxation: We allow the facilities to be located in the plane. This is a departure from the prior work, where the agents and the facilities lie in the same metric space~\cite{Barbera1993,Escoffier2011,Goel2023,Kim1984,Peters1992,Sui2015,Sui2013}. Intuitively, this relaxation allows us to more easily attain constant-strategyproofness at the expense of increasing our approximation ratio of optimal social cost.

\subsection{Our Contributions}

These relaxations enable our two main results. First, we present mechanism $\Mech_3$, a natural mechanism that is constant-approximate (i.e., attains a constant approximation ratio of optimal social cost) and constant-strategyproof. Informally, the mechanism
starts with an optimal solution on the line and then moves each facility away from the line by a distance inversely proportional to the number of agents it serves. Intuitively, lifting facilities off the real line weakens the ability of an agent to significantly reduce their own service cost through misreporting, while still keeping the social cost within a constant factor of optimal. This technique is reminiscent of ‘money burning’ ideas that limit manipulation at the cost of some efficiency to preserve truthful behavior. A characteristic of mechanism $\Mech_3$ is that small changes in agent locations can cause unbounded changes in the locations of facilities. To address this characteristic, we present mechanism $\Mech_4$, which is constant-approximate, constant-strategyproof, and constant-Lipschitz (i.e., Lipschitz continuous with a constant Lipschitz factor).

\subsection{Related Work}

Facility location is a central problem in mechanism design without money. There has been extensive research
that studies facility location under various settings such as fully
general metric spaces~\cite{Fernandes2013,Mahdian2006}, strictly
convex spaces \cite{PingzhongTang2020}, graph
metrics~\cite{Alon2010,Cheng2013,Church1978,Filimonov2021,Schumner2002},
and probabilistic metrics \cite{Auricchio2024,Klootwijk2019}. We
highlight results that are directly relevant to ours and defer the
rest to a survey of the field~\cite{Chan2021}.

\subsubsection{Facility Location in Higher Dimensions}

One relevant line of investigation is the study of facility location in higher dimensions~\cite{Barbera1993,Escoffier2011,Goel2023,Kim1984,Peters1992,Sui2015,Sui2013}.
For example, Barber\`a et al.~\cite{Barbera1993} generalize Moulin's single-facility result to higher dimension by selecting the facility's coordinate within each dimension independently in a ``median-like"
fashion. Sui et al.~\cite{Sui2013} generalize median mechanisms from
\cite{Barbera1993,Black1948,Moulin1980} into percentile mechanisms by determining facility locations using percentiles of ordered agent
location projections in each dimension. While these mechanisms are strategyproof, they show that for any dimension greater than one, no percentile mechanism attains an approximation ratio with respect to social
cost bounded by a function of $n$. 

\subsubsection{Strategyproofness Relaxations}
In recent years, there has been growing interest in exploring approximate notions of strategyproofness in order to enable trade-offs with other desired properties~\cite{Birrell2011,Dale2022,Ding2021,Lee2015,Lubin2012,Oomine2017,Schneider2017,Schvartzman2020,Sui2015}. One relaxation is to allow an agent to obtain a
multiplicative constant factor gain through
misreporting~\cite{Fukui2019,Istrate2022,Oomine2017}. This relaxation
closely parallels the incentive ratio in market
design~\cite{Chen2011}. Oomine et al.~\cite{Oomine2017} studied the
obnoxious facility game~\cite{Cheng2011} under this relaxation and
exhibited a trade-off between strategyproofness and the
approximation ratio. Specifically, with the multiplicative
constant factor of $\lambda$, they obtained the approximation
ratio of $1+\frac{2}{\lambda}$. The
multiplicative factor can be viewed as providing a bound on the
maximum possible gain an agent can obtain by misreporting, which is broadly aligned with the goal of designing mechanisms with limited manipulability
\cite{Chen2016,Ding2021,Hyafil2006,Schneider2017,Schvartzman2020}.

\subsubsection{Facility Reallocation}

The setting of facility reallocation motivates the design of Lipschitz continuous mechanism~\cite{Fotakis2021,Keijzer2022}. The facility reallocation
problem considers a sequence of agent preference profiles
over time with the optimization goal being to minimize both the social
cost at each time step and the facility reallocation
distances across time steps. Lipschitz continuous mechanisms allow us to bound the facility reallocation distances as a function of the preference profile
distances without any information about prior or future
time steps.

\section{Preliminaries}

For any nonnegative integer $k$, we write $[k]$ as a shorthand for
$\{1,\ldots,k\}$. 

In this work, any instance of the facility location problem involves a
set of two or more agents indexed starting at $1$. A profile $\Prof$ is a
vector of real values specifying the preferences of the agents on the
real line in nondecreasing order.  The nondecreasing requirement
reflects the fact that all of the mechanisms we consider are
anonymous. For any positive integer $n$, we define an $n$-profile as a
profile for an $n$-agent instance.  For any $n$-profile $\Prof$, we
define $\Mean{\Prof}$ as $\frac{1}{n}\sum_{i \in [n]}\Prof_i$.  A
profile $\Prof$ is \emph{trivial} if all of its components are equal;
otherwise, it is \emph{nontrivial}.

For any $n$-profile $\Prof$ and any $i$ in $[n]$, we define
$\Interval{\Prof}{i}$ as $(-\infty,\Prof_2]$ if $i=1$, as
$[\Prof_{n-1},\infty)$ if $i=n$, and as $[\Prof_{i-1},\Prof_{i+1}]$
otherwise. Observe that a real number $x$ belongs to $\Interval{\Prof}{i}$ if and only if $(\Prof_1,\ldots,\Prof_{i-1},x,\Prof_{i+1},\ldots,\Prof_n)$
is an $n$-profile. For any $n$-profile $\Prof$, we define $\Hood{\Prof}$ as the set of
all $n$-profiles $\Prof^*$ such that $|\{i\in[n]\mid\Prof_i\neq\Prof^*_i\}|\leq 1$. 
For any $n$-profile $\Prof$, any
$i$ in $[n]$, and any real number $x$, we define $\Subst{\Prof}{i}{x}$ as
the $n$-profile obtained from $\Prof$ by replacing $\Prof_i$
with $x$ and rearranging the resulting components in nondecreasing
order.

In this paper, we use the term ``point'' to refer to an element of
$\Reals^2$, and we write $\Point_1$ (resp., $\Point_2$) to denote the
$x$-coordinate (resp., $y$-coordinate) of a given point $\Point$. We
say that a point $\Point$ is \emph{1-D} if $\Point_2=0$.  For any real
number $x$ and any point $\Point$, we define $\Dist{x}{\Point}$ as
$\|(x,0)-p\|_1=|x-p_1|+|p_2|$.\footnote{Even though we focus on the Manhattan distance in this paper, all of our results also hold for Euclidean distance up to constant factors.}

A solution $\Soln$ for a given profile $\Prof$ is an ordered pair
$(\Soln_1,\Soln_2)$ of points specifying the locations of the two
facilities. We require that $\Soln_1\leq\Soln_2$, where the comparison
is performed lexicographically.  We sometimes refer to $\Soln_1$
(resp., $\Soln_2$) as the left (resp., right) facility.  We write
$\Soln_{i,j}$, where $i$ and $j$ belong to $\{1,2\}$, to refer to
component $j$ of $\Soln_i$.  We say that a solution $\Soln$ is
\emph{1-D} if $\Soln_1$ and $\Soln_2$ are each 1-D.\footnote{We adopt
the convention that the symbol $\SolnOne$ (rather than $\Soln$) is
used to refer to a solution that is guaranteed to be 1-D.}  For any
real number $x$ and any solution $\Soln$, we define $\Dist{x}{\Soln}$
as $\min(\Dist{x}{\Soln_1},\Dist{x}{\Soln_2})$.  In other words,
$\Dist{x}{\Soln}$ is equal to the distance between $x$ and its nearest
facility, which corresponds to the \emph{individual cost function} of
an agent with preference $x$.  For any two solutions $\Soln$ and
$\Soln^*$, we define $\|\Soln-\Soln^*\|_1$ as $\sum_{\ell \in
  \Set{1,2}}\|\Soln_\ell-\Soln^*_\ell\|_1$.

For any profile $\Prof$ and any solution $\Soln$, we define
$\CostSoln{\Prof}{\Soln}$ as $\sum_{i \in [n]}\Dist{\Prof_i}{\Soln}$.
In other words, we define the \emph{social cost} of a profile $\Prof$
and solution $\Soln$ as the sum of individual agent costs
$\Dist{\Prof_i}{\Soln}$.

\begin{lemma}
\label{lem:lipCostSoln}
Let $\Prof$ and $\Prof^*$ be $n$-profiles and let $\Soln$ be a
solution. Then $|\CostSoln{\Prof^*}{\Soln}-\CostSoln{\Prof}{\Soln}|$
is at most $\|\Prof^*-\Prof\|_1$.
\end{lemma}
\begin{proof}
For any $i$ in $[n]$,
$|\Dist{\Prof^*_i}{\Soln}-\Dist{\Prof_i}{\Soln}|$ is at most
$|\Prof^*_i-\Prof_i|$.
\end{proof}

In our analysis, we sometimes want to identify the set of agents
served by each facility under a given solution.  For any profile
$\Prof$ and any solution $\Soln$, we define
$\ServedSoln{\Prof}{1}{\Soln}$ as the set of all indices $i$ in
$\AgentSet$ such that
$\Dist{\Prof_i}{\Soln_1}\leq\Dist{\Prof_i}{\Soln_2}$, and we define
$\ServedSoln{\Prof}{2}{\Soln}$ as
$\AgentSet\setminus\ServedSoln{\Prof}{1}{\Soln}$. Note that if an
agent is equidistant from the two facilities, we consider it to be
served by the left facility.

We define the left-median of $s\geq 1$ real numbers $x_1\leq\cdots\leq
x_s$ as $x_{\lceil s/2\rceil}$. That is, if $s$ is odd, then the
left-median is the median, and if $s$ is even, then the left-median is
the lower-indexed of the two middle values. The following fact is attributable to Procaccia and Tennenholtz~\cite{Procaccia2013}.

\begin{fact}
\label{fact:single}
For any profile $\Prof$, the lexicographically first minimum-cost
single-facility solution locates the facility at the left-median of
$\Prof$.
\end{fact}

For any $n$-profile $\Prof$ and any $i$ in $[n-1]$, we define
$\Cand{\Prof}{i}$ as the 1-D solution that locates the left facility
at the left-median of the leftmost $i$ agents (i.e., at $\Prof_{\lceil
  i/2\rceil}$), and the right facility at the left-median of the
rightmost $n-i$ agents (i.e., at $\Prof_{\lceil(n+i)/2\rceil}$).  We
also define $\Cand{\Prof}{n}$ as the 1-D solution $((\Prof_{\lceil
  n/2\rceil},0),(\Prof_{\lceil n/2\rceil},0))$.  Solution
$\Cand{\Prof}{n}$ locates both facilities at the left-median of the
entire set of $n$ agents; note that the left facility is viewed as
serving the entire set of $n$ agents under this solution.

For any $n$-profile $\Prof$, we define $\Active{\Prof}$ as the set of
all $i$ in $[n]$ such that under solution $\Cand{\Prof}{i}$, the left
facility serves exactly $i$ agents.  For any $n$-profile $\Prof$, we
define $\Canonical{\Prof}$ as the lexicographically first minimum-cost
solution for $\Prof$ if $\Prof$ is nontrivial, and as
$\Cand{\Prof}{n}$ otherwise.

\begin{lemma}
\label{lem:active}
Let $\Prof$ be an $n$-profile and let $\SolnOne$ denote
$\Canonical{\Prof}$.  Then there is an index $i$ in $\Active{\Prof}$
such that $\SolnOne=\Cand{\Prof}{i}$.
\end{lemma}
\begin{proof}
If $\Prof$ is trivial, it is clear that $n$ belongs to
$\Active{\Prof}$ and $\SolnOne=\Cand{\Prof}{n}$.  For the remainder of
the proof, assume that $\Prof$ is nontrivial.  Thus $\SolnOne$ is the
lexicographically first minimum-cost solution for $\Prof$. Let $i$
denote $|\ServedSoln{\Prof}{1}{\SolnOne}|$.  Since $\Prof$ is
nontrivial and $\SolnOne$ is minimum-cost, we find that
$\SolnOne_{1,1}<\SolnOne_{2,1}$ and $i$ belongs to $[n-1]$. Since
$\SolnOne$ is the lexicographically first minimum-cost solution for
$\Prof$, we deduce from Fact~\ref{fact:single} that $\SolnOne_{1,1}$
is the left-median of $\Prof_1,\ldots,\Prof_i$ and $\SolnOne_{2,1}$ is
the left-median of $\Prof_{i+1},\ldots,\Prof_n$. Thus $i$ belongs to
$\Active{\Prof}$ and $\SolnOne=\Cand{\Prof}{i}$.
\end{proof}

Using Lemma~\ref{lem:active}, we deduce that for any $n$-profile
$\Prof$, there is exactly one $i$ in $[n]$ such that
$\Canonical{\Prof}=\Cand{\Prof}{i}$; we define this $i$ as
$\Index{\Prof}$.

For any profile $\Prof$, we define the \emph{optimal social cost}
$\Cost{\Prof}$ as $\CostSoln{\Prof}{\SolnOne}$ and
$\Served{\Prof}{\ell}$ as $\ServedSoln{\Prof}{\ell}{\SolnOne}$ for all
$\ell$ in $\{1,2\}$, where $\SolnOne$ denotes $\Canonical{\Prof}$.

\begin{lemma}
\label{lem:lipCost}
Let $\Prof$ and $\Prof^*$ be $n$-profiles. Then
$|\Cost{\Prof^*}-\Cost{\Prof}|\leq \|\Prof^*-\Prof\|_1$.
\end{lemma}
\begin{proof}
By symmetry, it is sufficient to establish that
$\Cost{\Prof^*}-\Cost{\Prof}\leq \|\Prof^*-\Prof\|_1$.  Let $\SolnOne$
(resp., $\SolnOne^*$) denote the canonical solution for profile
$\Prof$ (resp., $\Prof^*$).  We have
$\Cost{\Prof^*}\leq\CostSoln{\Prof^*}{\SolnOne}
\leq\Cost{\Prof}+\|\Prof^*-\Prof\|_1$, where the second inequality
follows from Lemma~\ref{lem:lipCostSoln}.
\end{proof}

A mechanism $\Mech$ maps any given profile to a solution.  We say that
a mechanism is 1-D if it only produces 1-D solutions.

For any mechanism $\Mech$ and any function $\FactorLetter: \mathbb{N}
\rightarrow \Reals^+$, we define the following terms:
\begin{enumerate}
\item $\Mech$ is \emph{$\Factor{n}$-approximate} if
  $\CostSoln{\Prof}{\Mech(\Prof)} \le \Factor{n} \Cost{\Prof}$ for all
  $n$-profiles $\Prof$;
\item $\Mech$ is \emph{$\Factor{n}$-strategyproof} if
  $\Dist{\Prof_i}{\Mech(\Prof)} \le \Factor{n}
  \Dist{\Prof_i}{\Mech(\Subst{\Prof}{i}{x})}$ for all $n$-profiles
  $\Prof$, all $i$ in $[n]$, and all real numbers $x$;
\item $\Mech$ is \emph{$\Factor{n}$-Lipschitz} if
  $\|\Mech(\Prof^*)-\Mech(\Prof)\|_1 \le \Factor{n}
  \|\Prof^*-\Prof\|_1$ for all $n$-profiles $\Prof$ and $\Prof^*$.
\end{enumerate}

As discussed in \cref{sec:introduction}, our interest is in achieving
constant-factor bounds (i.e., independent of $n$). We
say a mechanism is constant-approximate (resp. constant-strategyproof,
constant-Lipschitz) if the factor $\Factor{n}$ is $O(1)$.  Note that
when $\Factor{n} = 1$, $\Factor{n}$-strategyproof is strategyproof and
when $\Factor{n}$ is a constant, $\Factor{n}$-strategyproof corresponds to
constant-strategyproof defined in \cite{Oomine2017}.

\section{Facility Location in $\mathbb{R}$}

\subsection{Mechanism $\Mech_1$}

We define our first mechanism $\MechOpt$ as follows: For
any profile $\Prof$, $\MechOpt(\Prof)$ is equal to $\Canonical{\Prof}$.

By construction, it is clear to see that this mechanism attains the optimal social cost and therefore is constant-approximate with an approximation ratio of $1$. However, the following counter-example shows that this mechanism is neither constant-Lipschitz nor constant-strategyproof.

\begin{lemma}
\label{negexample}
    Mechanism $\MechOpt$ is not constant-strategyproof and is not constant-Lipschitz.
\end{lemma}
\begin{proof}
Let $\lambda$ be an arbitrarily large positive real number, let $z$ and $\eps$ be positive real numbers such that $z > 2\eps\lambda$, let $\Prof$ denote the profile $(-z,0,0,z+\eps)$, let $\SolnOne$ denote the solution $\MechOpt(\Prof)=((0,0),(z+\eps,0))$, let $\Prof^*$ denote the profile $(-z-2\eps,0,0,z+\eps)$, and let $\SolnOne^*$ denote the solution $\MechOpt(\Prof^*) = ((-z-2\eps,0),(0,0))$.
Observe that only agent $1$ changed their report (from $\Prof$ to $\Prof^*$) and only changed it by $2\eps$. 
We have $$d(\Prof_1,\SolnOne)=z \geq 2\eps\lambda=\lambda d(\Prof_1^*,\SolnOne^*).$$ Hence $\MechOpt$ is not constant-strategyproof. Furthermore,
$$\frac{\|\SolnOne-\SolnOne^*\|_1}{\|\Prof-\Prof^*\|_1}=\frac{2z+3\eps}{2\eps},$$
which goes to infinity as $z\to\infty$. Hence $\MechOpt$ is not constant-Lipschitz.
\end{proof}

The linear bound on the approximation ratio of optimal social cost for
deterministic, strategyproof, mechanisms provided by \cite{Lu2010} can
trivially be extended to constant-strategyproof mechanisms to show
that, in fact, no 1-D, constant-strategyproof, mechanism can attain a
constant approximation ratio of optimal social cost. However, such a
mechanism can be constant-Lipschitz as we show with the mechanism
$\MechLip$ (\cref{sec:mechanism2}). We use the following lemma in
our analysis of $\MechLip$.

\begin{lemma}
	\label{lem:median}
	Let $\Prof$ be a profile and let $\SolnOne$ denote
	$\MechOpt(\Prof)$. Then $|\{i\mid\Prof_i\leq\SolnOne_{1,1}\}|$ is at
	least $\lceil|\Served{\Prof}{1}|/2\rceil$ and
	$|\{i\mid\Prof_i\geq\SolnOne_{2,1}\}|$ is at least
	$\lceil|\Served{\Prof}{2}|/2\rceil$.
\end{lemma}
\begin{proof}
	Follows from Lemma~\ref{lem:active}.
\end{proof}

\subsection{Mechanism $\Mech_2$}
\label{sec:mechanism2}

Below we define our second mechanism $\MechLip$. Before doing so, we provide some useful definitions. 

For any $n$-profile $\Prof$, let $\Helper{\Prof}{1}$ denote the unique
real number $x$ greater than or equal to $\Prof_1$ such that 
$\sum_{i \in [n]}\max(0,x-\Prof_i)=\Cost{\Prof}$, let $\Fac{\Prof}{1}$ denote
$\min(\Helper{\Prof}{1},\Mean{\Prof})$. Similarly, let $\Helper{\Prof}{2}$ denote
the unique real number $x$ less than or equal to $\Prof_n$ such that
$\sum_{i \in [n]}\allowbreak\max(0,\Prof_i-x)=\Cost{\Prof}$, and let
$\Fac{\Prof}{2}$ denote $\max(\Helper{\Prof}{2},\Mean{\Prof})$.
It is easy to argue that $\Fac{\Prof}{1}=\Mean{\Prof}$ if and
only if $\Fac{\Prof}{2}=\Mean{\Prof}$.

The reason that we introduce $\Fac{\Prof}{1}$ and $\Fac{\Prof}{2}$ is because it is possible that $\Helper{\Prof}{1} \ge \Helper{\Prof}{2}$. For example, let $\Prof$ denote the profile $(-2, -1, 0, 0, 1, 1)$. It is straightforward to check that $\Helper{\Prof}{1} = 0$ and $\Helper{\Prof}{2} = -1/4$. However, for all $\ell$ belonging to $\Set{1,2}$, $\FacB{\Prof}{\ell} = -1/6$.

We define mechanism~$\MechLip$ as follows:
For any profile $\Prof$, $\MechLip(\Prof)$ is the 1-D solution
$((\Fac{\Prof}{1},0),(\Fac{\Prof}{2},0))$.

We now show that mechanism $\MechLip$ is constant-approximate and establish properties to be used in \cref{sec:mechanism4} to prove that mechanism $\MechLipTwo$ is constant-Lipschitz. Lemma~\ref{lem:lipMain}, which proves that mechanism $\MechLipTwo$ is constant-Lipschitz, also establishes that mechanism~$\MechLip$ is constant-Lipschitz.

We begin with the following lemma towards the goal of showing that mechanism $\MechLip$ is constant-approximate.

\begin{lemma}
\label{lem:nest}
Let $\Prof$ be a profile and let $\SolnOne$ denote $\MechOpt(\Prof)$.
Then
\[
\SolnOne_{1,1}\leq\Fac{\Prof}{1}\leq\Fac{\Prof}{2}\leq\SolnOne_{2,1}.
\]
\end{lemma}
\begin{proof}
Immediate from the definitions.
\end{proof}

The following lemma establishes that mechanism $\MechLip$ is constant-approximate.

\begin{lemma}
\label{lem:costLip}
Let $\Prof$ be an $n$-profile, let $\SolnOne$ denote $\MechLip(\Prof)$. Then $\CostSoln{\Prof}{\SolnOne}\leq 3\Cost{\Prof}$.
\end{lemma}

\begin{proof}
Let $X$ denote $\{i\in[n]\mid\Prof_i\leq\Fac{\Prof}{1}\}$, $Y$
denote $\{i\in[n]\mid\Prof_i\geq\Fac{\Prof}{2}\}$, and $Z$ denote
$[n]\setminus(X\cup Y)$.  
The definition of mechanism~$\MechLip$
implies $\sum_{i\in X}\Dist{\Prof_i}{\SolnOne}\leq\Cost{\Prof}$ and
$\sum_{i\in  Y}\Dist{\Prof_i}{\SolnOne}\leq\Cost{\Prof}$. Lemma~\ref{lem:nest}
implies $\sum_{i\in Z}\Dist{\Prof_i}{\SolnOne}\leq\Cost{\Prof}$. The claim of the lemma follows.
\end{proof}

Again, given a trivial extension to the lower bound provided by \cite{Lu2010}, we know that mechanism $\MechLip$ cannot be constant-strategy proof because it is a 1-D mechanism that attains a constant approximation ratio of optimal social cost. 

We now move on to establish  properties of
mechanism~$\MechLip$ that are used for analyzing its 2-D
generalization, mechanism $\MechLipTwo$, in Section \ref{sec:mechanism4}.

We begin with some useful definitions.  For any profile $\Prof$, let
$\Gap{\Prof}$ denote $\Fac{\Prof}{2}-\Fac{\Prof}{1}$, the distance between the two facilities on the $x$-axis.  For any profile
$\Prof$ such that $\Gap{\Prof}>0$ and any $i$ in $[n]$, we define $\Frac{\Prof}{2}{i}$ as

$$\frac{\max(0,\min(\Fac{\Prof}{2},\Prof_i)-\Fac{\Prof}{1})}{\Gap{\Prof}}$$
and $\Frac{\Prof}{1}{i}$ as $1-\Frac{\Prof}{2}{i}$.  
For any
profile $\Prof$ such that $\Gap{\Prof}>0$ and any $\ell$ in $\{1,2\}$,
we define $\Weight{\Prof}{\ell}$ as $\sum_{1\leq i\leq
  n}\Frac{\Prof}{\ell}{i}$ 
  and $\CostDensity{\Prof}{\ell}$
as $\Cost{\Prof}/\Weight{\Prof}{\ell}$. 

\begin{lemma}
\label{lem:weightLowerBound}
For any $\ell$ in $\{1,2\}$ and any profile $\Prof$ such that
$\Gap{\Prof}>0$, we have
\[
\Weight{\Prof}{\ell}\geq|\Served{\Prof}{\ell}|/2.
\]
\end{lemma}
\begin{proof}
Immediate from Lemmas ~\ref{lem:median} and ~\ref{lem:nest}.
\end{proof}

\begin{lemma}
\label{lem:weightLowerBoundLip}
Let $\Prof$ be a profile such that $\Gap{\Prof}>0$, let $\ell$ belong
to $\{1,2\}$, and let $\SolnOne$ denote $\MechLip(\Prof)$. Then
\[
\Weight{\Prof}{\ell}\geq|\ServedSoln{\Prof}{\ell}{\SolnOne}|/2.
\]
\end{lemma}
\begin{proof}
Observe that each agent in $\ServedSoln{\Prof}{\ell}{\SolnOne}$
contributes at least $1/2$ to $\Weight{\Prof}{\ell}$.
\end{proof}

\begin{lemma}
\label{lem:weightUpperBoundHelper}
Let $\Prof$ be an $n$-profile such that $\Gap{\Prof}>0$ and let $\ell$ belong
to $\{1,2\}$.  Then
\[
|\Served{\Prof}{\ell}|
+\Cost{\Prof}/\Gap{\Prof}\geq\Weight{\Prof}{\ell}.
\]
\end{lemma}
\begin{proof}
Below we address the case $\ell=2$. A similar argument holds for
the case $\ell=1$.

Let $\SolnOne$ denote $\MechOpt(\Prof)$ and let 
$x$ denote $(\SolnOne_{1,1}+\SolnOne_{2,1})/2$. We have
\begin{eqnarray*}
\Cost{\Prof}
 & \geq &
\sum_{i\in [n]:\Prof_i\leq x}\max(0,\Prof_i-\SolnOne_{1,1}) \\
 & \geq & 
\sum_{i\in [n]:\Prof_i\leq x}\max(0,\Prof_i-\Fac{\Prof}{1}) \\
 & \geq & 
\sum_{i\in [n]:\Prof_i\leq x}\max(0,\min(\Fac{\Prof}{2},\Prof_i)-\Fac{\Prof}{1}) \\
 & = &
\Gap{\Prof}\sum_{i\in [n]:\Prof_i\leq x}\Frac{\Prof}{2}{i}.
\end{eqnarray*}
Moreover,
\[
|\Served{\Prof}{2}|=|\{i\in[n]\mid\Prof_i>x\}|
\geq\sum_{i\in[n]:\Prof_i>x}\Frac{\Prof}{2}{i}.
\]
Since
$\Weight{\Prof}{2}=\sum_{1\leq i\leq n}\Frac{\Prof}{2}{i}$, we
conclude that the claim of the lemma holds for $\ell=2$.
\end{proof}

\begin{lemma}
\label{lem:weightUpperBound}
Let $\Prof$ be a profile such that $\Gap{\Prof}>0$, let $\ell$ belong
to $\{1,2\}$, and assume that $\Gap{\Prof}\geq
2\CostDensity{\Prof}{\ell}$. Then
\[
\Weight{\Prof}{\ell}\leq 2|\Served{\Prof}{\ell}|.
\]
\end{lemma}
\begin{proof}
The inequality $\Gap{\Prof}\geq 2\CostDensity{\Prof}{\ell}$ implies
$\Cost{\Prof}/\Gap{\Prof}$ is at most
$\Weight{\Prof}{\ell}/2$. Hence the claim follows from
Lemma~\ref{lem:weightUpperBoundHelper}.
\end{proof}

\section{Facility Location in $\mathbb{R}^2$}
\label{sec:twoDim}

We now move on to defining 2-D mechanisms, both of which are
constant-strategyproof and constant-approximate, one of which is also
constant-Lipschitz.

\subsection{Mechanism $\Mech_3$}
\label{sec:mechanism3}

Consider the following 2-D generalization of mechanism~$\MechOpt$,
which we refer to as mechanism~$\MechOptTwo$.  Given a profile
$\Prof$, and letting $\SolnOne$ denote $\MechOpt(\Prof)$, we define
$$\MechOptTwo(\Prof)=((\SolnOne_{1,1},\Cost{\Prof}/|\Served{\Prof}{1}|),(\SolnOne_{2,1},\Cost{\Prof}/|\Served{\Prof}{2}|)).$$ In other words, $\MechOptTwo$ is the generalization of $\MechOpt$ in which each facility is vertically backed off of the $x$-axis by a distance equal to the optimal social cost divided by the number of agents served by that facility in $\MechOpt$.\footnote{If $\Prof$ is trivial, then we set $\Soln_{2,2}=0$.}

The following lemma, which establishes that $\MechOptTwo$ is constant-approximate, is straightforward to prove.

\begin{lemma}
\label{lem:costOptTwo}
For any profile $\Prof$, we have $\CostSoln{\Prof}{\MechOptTwo(\Prof)}\leq
3\Cost{\Prof}$.
\end{lemma}

We now present a sequence of lemmas to be used to establish that $\MechOptTwo$ is constant-strategyproof.

In \cref{lemma:fnhc-oc}, \cref{lem:hard-case}, and \cref{theorem:c1-c2-lambda}, let $\eps$ denote $(2-\sqrt{3})/3$.
For this choice of $\eps$, one may verify that
$(\frac{1}{3}-\eps)(1 - 3\eps) = 2\eps$.

\begin{lemma}\label{lemma:fnhc-oc}
	Let $\Prof$ be an $n$-profile, let $i$ belong to $[n]$, let $x$ be a
	real number, let $\ProfLie$ denote $\Subst{\Prof}{i}{x}$, 
	let $\Soln$ denote $\MechOptTwo(\Prof)$,
	let $\SolnLie$ denote $\MechOptTwo(\ProfLie)$, let $\SolnOne$ denote
	$((\Soln_{1,1},0), (\Soln_{2,1},0))$, let $\SolnOneLie$ denote
	$((\SolnLie_{1,1},0), (\SolnLie_{2,1},0))$, and assume that
	$\Dist{\Prof_i}{\SolnOneLie} \le \eps \Dist{\Prof_i}{\Soln}$. Then
	$\Cost{\ProfLie} \ge (\frac{1}{3} - \eps)
	\CostSoln{\Prof}{\Soln}$.
\end{lemma}
\begin{proof}
	We have
	\begin{eqnarray*}
		\Cost{\ProfLie}
		& = &
		\CostSoln{\Prof}{\SolnOneLie} + \Dist{x}{\SolnOneLie}
		- \Dist{\Prof_i}{\SolnOneLie} \\
		& \ge &
		\Cost{\Prof} + 0 -\eps \Dist{\Prof_i}{\Soln} \\
		& \ge &
		\frac{\CostSoln{\Prof}{\Soln}}{3} - \eps \Dist{\Prof_i}{\Soln} \\
		& \ge &
		\left(\frac{1}{3} - \eps \right) \CostSoln{\Prof}{\Soln},
	\end{eqnarray*}
	where the second inequality follows from \cref{lem:costOptTwo}. 
\end{proof}

\begin{lemma}\label{lem:hard-case}
	Let $\Prof$ be an $n$-profile, let $i$ belong to $[n]$, let $x$ be a
	real number, let $\ProfLie$ denote $\Subst{\Prof}{i}{x}$, 
	let $\Soln$ denote $\MechOptTwo(\Prof)$,
	let $\SolnLie$ denote $\MechOptTwo(\ProfLie)$, let $\SolnOneLie$ denote
	$((\SolnLie_{1,1},0),(\SolnLie_{2,1},0))$, let $\ell$ belong to
	$\Set{1,2}$, and assume that $\Dist{\Prof_i}{\SolnOneLie_\ell} \le
	\eps \Dist{\Prof_i}{\Soln}$. Then $\SolnLie_{\ell,2} \ge \eps
	\Dist{\Prof_i}{\Soln}$.
\end{lemma}
\begin{proof}
    Let $u$ be a bijection from $[n]$ to $[n]$ such that
	$\ProfLie_{u(j)}=\Prof_j$ for all $j$ in $[n]\setminus\{i\}$ and
	$\ProfLie_{u(i)}=x$, let $J$ denote $\{j\in[n]\mid
	u(j)\in\ServedSoln{\ProfLie}{\ell}{\SolnOneLie}\}$, let $A$ denote
	$\{j\in J\mid\Dist{\ProfLie_{u(j)}}{\SolnOneLie}\le
	2\eps\Dist{\Prof_i}{\Soln}\}$, and let $B$ denote $J\setminus A$.
	
    For any $j$ in $A$, we claim that $\Dist{\Prof_j}{\Soln}\ge
	(1-3\eps)\Dist{\Prof_i}{\Soln}$.  If $i=j$, the claim holds trivially.
	Suppose $i\neq j$.  Since
	$|\SolnOneLie_\ell-\Prof_i|=\Dist{\Prof_i}{\SolnOneLie_\ell}
	\le\eps\Dist{\Prof_i}{\Soln}$ and
	$|\SolnOneLie_\ell-\ProfLie_{u(j)}|=|\SolnOneLie_\ell-\Prof_j|\le
	2\eps \Dist{\Prof_i}{\Soln}$, the triangle inequality implies
	$|\Prof_i-\Prof_j|\le 3\eps\Dist{\Prof_i}{\Soln}$.  Since
	$\Dist{\Prof_j}{\Soln}$ is at least
	$\Dist{\Prof_i}{\Soln}-|\Prof_i-\Prof_j|$, the claim follows.
	
	The preceding claim implies $\CostSoln{\Prof}{\Soln} \ge \sum_{j \in
		A} \Dist{\Prof_j}{\Soln} \ge (1-
	3\eps)|A|\cdot\Dist{\Prof_i}{\Soln}$. 
    Since $\Dist{\Prof_i}{\SolnOneLie} \le \Dist{\Prof_i}{\SolnOneLie_\ell} \le  \eps \Dist{\Prof_i}{\Soln}$,
    Lemma~\ref{lemma:fnhc-oc}
	implies 
    \[
    \Cost{\ProfLie} \ge \left(\frac{1}{3}-\eps\right)
	(1- 3\eps)|A|\cdot\Dist{\Prof_i}{\Soln}.
    \]
    The definition of $B$
	implies
    $\Dist{\ProfLie_{u(j)}}{\SolnOneLie} \ge 2\eps
	\Dist{\Prof_i}{\Soln}$ for all $j$ in $B$ and hence 
	$\Cost{\ProfLie} \ge \sum_{j \in B}
	\Dist{\ProfLie_{u(j)}}{\SolnOneLie} \geq 2\eps |B| \cdot
	\Dist{\Prof_i}{\Soln}$. Thus $$\Cost{\ProfLie} \ge
	\max\left[\left(\frac{1}{3}-\eps\right) (1- 3\eps)|A|, 2\eps
	|B|\right]\allowbreak\Dist{\Prof_i}{\Soln}.$$  
    
    Recall that $(\frac{1}{3}-\eps)(1 -
	3\eps) = 2\eps$.  Thus
	\[
	\Cost{\ProfLie} \ge
	2\eps\max(|A|,|B|)\Dist{\Prof_i}{\Soln} \ge
	\eps|J|\cdot\Dist{\Prof_i}{\Soln}
	\]
	and hence $\SolnLie_{\ell,2} = \Cost{\ProfLie}/|J| \ge
	\eps\Dist{\Prof_i}{\Soln}$.
\end{proof}

\begin{lemma}\label{theorem:c1-c2-lambda}
Mechanism $\MechOptTwo$ is $\frac{1}{\eps}$-strategyproof.
\end{lemma}
\begin{proof}
Let $\Prof$ be an $n$-profile, let $i$ belong to $[n]$, let $x$ be a
real number, let $\ProfLie$ denote $\Subst{\Prof}{i}{x}$, let $\Soln$
denote $\MechOptTwo(\Prof)$, let $\SolnLie$ denote $\MechOptTwo(\ProfLie)$, and
let $\SolnOneLie$ denote $((\SolnLie_{1,1},0),
(\SolnLie_{2,1},0))$. We need to prove that $\eps\Dist{\Prof_i}{\Soln} \le \Dist{\Prof_i}{\SolnLie}$. 
Assume for the sake of contradiction that
$\Dist{\Prof_i}{\SolnLie}<\eps\Dist{\Prof_i}{\Soln}$ and let $\ell$
be an element of $\Set{1,2}$ such that $i$ belongs to
$\ServedSoln{\Prof}{\ell}{\SolnLie}$. Then
\[
\Dist{\Prof_i}{\SolnOneLie_\ell} \leq
\Dist{\Prof_i}{\SolnLie_\ell} =
\Dist{\Prof_i}{\SolnLie} <
\eps\Dist{\Prof_i}{\Soln}
\]
and hence Lemma~\ref{lem:hard-case} implies
$\SolnLie_{\ell,2}\geq\eps\Dist{\Prof_i}{\Soln}$. It follows
that $\Dist{\Prof_i}{\SolnLie}\geq\eps\Dist{\Prof_i}{\Soln}$,
a contradiction.
\end{proof}

\begin{theorem}
    \label{M3_Thm}
    Mechanism $\MechOptTwo$ is constant-approximate and constant-strategyproof.
\end{theorem}
\begin{proof}
    The theorem follows immediately from \cref{lem:costOptTwo} and \cref{theorem:c1-c2-lambda}.
\end{proof}

The proof of \cref{negexample} can be easily adapted to show that $\MechOptTwo$ is neither strategyproof nor constant-Lipschitz.

\subsection{Mechanism $\Mech_4$}
\label{sec:mechanism4}

Our final mechanism, $\MechLipTwo$, is a 2-D generalization of mechanism~$\MechLip$. In order to define $\MechLipTwo$, we first introduce the following definitions for any profile $\Prof$ and any $\ell$ in $\{1,2\}$:
$\BackoffHelper{\Prof}{\ell}$ denotes $0$ if $\Gap{\Prof}=0$ and
$\min(\Gap{\Prof},2\CostDensity{\Prof}{\ell})$ otherwise; $\Backoff{\Prof}{\ell}$ denotes
$\max(8\Cost{\Prof}/n,\BackoffHelper{\Prof}{\ell})$.
 
Given a profile $\Prof$, and
letting $\SolnOne$ denote $\MechLip(\Prof)$, we define
$\MechLipTwo(\Prof)$ as the solution $\Soln$ such that
$\Soln_{\ell,1}=\SolnOne_{\ell,1}$ and
$\Soln_{\ell,2}=\Backoff{\Prof}{\ell}$ for all $\ell$ in $\{1,2\}$.

\begin{lemma}
\label{lem:costLipTwo}
For any $n$-profile $\Prof$, we have $\CostSoln{\Prof}{\MechLipTwo(\Prof)}\leq
15\Cost{\Prof}$.
\end{lemma}
\begin{proof}
Let $\SolnOne$ denote $\MechLip(\Prof)$. 
Observe that
\begin{equation}
\label{eqn:costLipTwo}
\CostSoln{\Prof}{\MechLipTwo(\Prof)}\leq\CostSoln{\Prof}{\SolnOne}
+\sum_{\ell\in\{1,2\}}\Backoff{\Prof}{\ell}|\ServedSoln{\Prof}{\ell}{\SolnOne}|.
\end{equation}
Since $\CostSoln{\Prof}{\SolnOne}\leq 3\Cost{\Prof}$ by
Lemma~\ref{lem:costLip}, it is sufficient to prove that the sum
in Equation~(\ref{eqn:costLipTwo}) is at most $12\Cost{\Prof}$.  If
$\Gap{\Prof}=0$, then $\Backoff{\Prof}{\ell}=8\Cost{\Prof}/n$ for all
$\ell$ in $\{1,2\}$ and hence this sum is equal to $8\Cost{\Prof}$.

For the remainder of the proof, assume that $\Gap{\Prof}>0$.  For any
$\ell$ in $\{1,2\}$, we have 
\begin{eqnarray*}
\Backoff{\Prof}{\ell}|\ServedSoln{\Prof}{\ell}{\SolnOne}|
& \leq &
\max(8\Cost{\Prof}/n,2\CostDensity{\Prof}{\ell})
|\ServedSoln{\Prof}{\ell}{\SolnOne}|\\
& \leq &
\max(8\Cost{\Prof}/n,4\Cost{\Prof}/|\ServedSoln{\Prof}{\ell}{\SolnOne}|)
|\ServedSoln{\Prof}{\ell}{\SolnOne}|\\
& = &
\max(8\Cost{\Prof}|\ServedSoln{\Prof}{\ell}{\SolnOne}|/n,4\Cost{\Prof}),
\end{eqnarray*}
where the first inequality follows from the definition of
$\Backoff{\Prof}{\ell}$ and the second inequality follows from
Lemma~\ref{lem:weightLowerBoundLip}.  Thus the sum in
Equation~(\ref{eqn:costLipTwo}) is at most $12\Cost{\Prof}$.
\end{proof}

We next define a core lemma which we will use to relate the constant-strategyproofness of mechanism $\MechOptTwo$ to that of mechanism $\MechLipTwo$.

\begin{lemma}\label{lem:oc-relate}
Let $\Mech$ and $\Mech^*$ be mechanisms, let $a_1$ and $a_2$ be
positive real numbers such that $a_1\Dist{x}{\Mech(\Prof)} \le
\Dist{x}{\Mech^*(\Prof)} \le a_2\Dist{x}{\Mech(\Prof)}$ for all
profiles $\Prof$ and all real numbers $x$, and assume that $M$ is
$\lambda$-strategyproof. Then $M^*$ is $\lambda a_2/a_1$-strategyproof.
\end{lemma}
\begin{proof}
Let $\Prof$ be an $n$-profile, let $i$ belong to $[n]$,
let $x$ be a real number, and let $\ProfLie$ denote
$\Subst{\Prof}{i}{x}$. We have 

\begin{eqnarray*}
\Dist{\Prof_i}{\Mech^*(\Prof)} 
&\le& a_2\cdot\Dist{\Prof_i}{\Mech(\Prof)} \\
&\le& \lambda a_2\cdot\Dist{\Prof_i}{\Mech(\ProfLie)} \\
&\le& (\lambda a_2/a_1)\Dist{\Prof_i}{\Mech^*(\ProfLie)}. 
\end{eqnarray*}
\end{proof}

We next show that mechanism $\MechLipTwo$ is constant-strategyproof. The following simple fact is used in the proof of
Lemma~\ref{thm:spLip} below.

\begin{fact}
\label{fact:spHelper}
Let $p$ and $q$ be two points such that $p_2$ and $q_2$ are positive
and let $z$ be a real number such that $|p_1-q_1|\leq z$. Then
\[
\min\left(1,\frac{p_2}{q_2+z}\right)
\leq
\frac{\Dist{x}{p}}{\Dist{x}{q}}
\leq
\max\left(1,\frac{p_2+z}{q_2}\right)
\]
for all real numbers $x$.
\end{fact}

\begin{lemma}
\label{thm:spLip}
Mechanism~$\MechLipTwo$ is constant-strategyproof.
\end{lemma}
\begin{proof}
Since mechanism~$\MechOptTwo$ is constant-strategyproof,
Lemma~\ref{lem:oc-relate} implies it is sufficient to exhibit positive
real numbers $a_1$ and $a_2$ such that, for all profiles $\Prof$ and all real numbers $x$,
$a_1\Dist{x}{\MechOptTwo(\Prof)} \le \Dist{x}{\MechLipTwo(\Prof)}
\le a_2\Dist{x}{\MechOptTwo(\Prof)}$.

Let $\Prof$ be an $n$-profile, let $x$ be a real number, let
$\SolnOne$ denote $\MechOpt(\Prof)$, let $\SolnOne^*$ denote
$\MechLip(\Prof)$, let $\Soln$ denote $\MechOptTwo(\Prof)$, let
$\Soln^*$ denote $\MechLipTwo(\Prof)$, and let $s_\ell$ denote
$|\Served{\Prof}{\ell}|$ for all $\ell$ in $\{1,2\}$.  Thus
$\Soln_{\ell,2}=\Cost{\Prof}/s_\ell$.

If $\Cost{\Prof}=0$, it is easy to see that $\Soln^*=\Soln$ and
hence $\Dist{x}{\Soln^*}=\Dist{x}{\Soln}$. For the remainder of
the proof, we assume that $\Cost{\Prof}>0$. It follows that
$\Soln_{\ell,2}$, $\Soln^*_{\ell,2}$, and $s_\ell$ are positive for
all $\ell$ in $\{1,2\}$. We assume without loss of generality that
$s_1\geq n/2$; the case $s_2\geq n/2$ is symmetric.

We claim that $\Soln^*_{1,2}=8\Cost{\Prof}/n$. If $\Gap{\Prof}\leq
8\Cost{\Prof}/n$, the claim is immediate from the definition of
mechanism~$\MechLipTwo$. Otherwise, $\Gap{\Prof}>0$ and since $s_1\geq
n/2$, \cref{lem:weightLowerBound} implies $\Weight{\Prof}{1}\geq n/4$,
which in turn implies $2\CostDensity{\Prof}{1}\leq
8\Cost{\Prof}/n$, so the claim is once again immediate from the
definition of mechanism~$\MechLipTwo$.

We claim that $|\SolnOne^*_{\ell,1}-\SolnOne_{\ell,1}|\leq
2\Cost{\Prof}/s_\ell$ for all $\ell$ in $\{1,2\}$. We prove this claim
for the case $\ell=1$; the case $\ell=2$ is symmetric.  By
\cref{lem:nest}, we need to prove that $\SolnOne^*_{1,1}\leq z$ where
$z$ denotes $\SolnOne_{1,1}+2\Cost{\Prof}/s_1$. By \cref{lem:median},
there are at least $s_1/2$ agents at or to the left of
$\SolnOne_{1,1}$. Thus $\sum_{i\in[n]:\Prof_i\leq z}z-\Prof_i$ is at
least $\Cost{\Prof}$. Hence the definition of mechanism $\MechLip$ implies
$\SolnOne^*_{1,1}\leq z$, as required.

Given that $n/2\leq s_1\leq n$ and using Fact~\ref{fact:spHelper} with $p_2$ equal to $8\Cost{\Prof}/n$,
$q_2$ equal to $\Cost{\Prof}/s_1$, and $z$ equal to
$2\Cost{\Prof}/s_1$, we have
\begin{eqnarray*}
\frac{\Dist{x}{\Soln_1^*}}{\Dist{x}{\Soln_1}}
& \leq & \max\left(1,\max_{n/2\leq s_1\leq n}
\frac{8\Cost{\Prof}/n+2\Cost{\Prof}/s_1}{\Cost{\Prof}/s_1}\right)\\
& = & \max(1,\max_{n/2\leq s_1\leq n}8s_1/n+2)\\
& = & 10.
\end{eqnarray*}
We can use Fact~\ref{fact:spHelper} in a similar manner to derive a lower bound
of $1$ for the ratio $\Dist{x}{\Soln_1^*}/\Dist{x}{\Soln_1}$. In
summary, we have
\begin{equation}
\label{eqn:ratioLeft}
1\leq\Dist{x}{\Soln_1^*}/\Dist{x}{\Soln_1}\leq 10.
\end{equation}

Case~1:
$\Gap{\Prof}\geq2\CostDensity{\Prof}{2}$. Lemmas~\ref{lem:weightLowerBound}
and~\ref{lem:weightUpperBound} imply $\Weight{\Prof}{2}/2\leq s_2\leq
2\Weight{\Prof}{2}$.  Since $|\SolnOne^*_{2,1}-\SolnOne_{2,1}|\leq
2\Cost{\Prof}/s_2$, we deduce that
$|\SolnOne^*_{2,1}-\SolnOne_{2,1}|\leq 4\CostDensity{\Prof}{2}$. Below
we use two subcases to establish that
\begin{equation}
\label{eqn:ratioRightA}
1/3\leq\Dist{x}{\Soln_2^*}/\Dist{x}{\Soln_2}\leq 12.
\end{equation}
\cref{eqn:ratioLeft,eqn:ratioRightA} together imply
$1/3\leq\Dist{x}{\Soln^*}/\Dist{x}{\Soln}\leq 12$. 

Case~1.1: $\Weight{\Prof}{2}\geq n/4$. Thus
$\Soln_{2,2}^*=8\Cost{\Prof}/n$ and $n/8\leq s_2\leq n/2$.  Hence
$2\Cost{\Prof}/n\leq\Soln_{2,2}\leq 8\Cost{\Prof}/n$ and
$|\SolnOne^*_{2,1}-\SolnOne_{2,1}|\leq 16\Cost{\Prof}/n$.  Using
Fact~\ref{fact:spHelper} with $p_2$ equal to $8\Cost{\Prof}/n$, $q_2$
between $2\Cost{\Prof}/n$ and $8\Cost{\Prof}/n$, and $z$ equal to
$16\Cost{\Prof}/n$, we find that \cref{eqn:ratioRightA} holds.

Case~1.2: $\Weight{\Prof}{2}\leq n/4$.  Thus
$\Soln_{2,2}^*=2\CostDensity{\Prof}{2}$. Since $\Weight{\Prof}{2}/2\leq
s_2\leq 2\Weight{\Prof}{2}$, we have
$\CostDensity{\Prof}{2}/2\leq\Soln_{2,2}\leq 2\CostDensity{\Prof}{2}$.
Using Fact~\ref{fact:spHelper} with $p_2$ equal to
$2\CostDensity{\Prof}{2}$, $q_2$ between $\CostDensity{\Prof}{2}/2$
and $2\CostDensity{\Prof}{2}$, and $z$ equal to
$4\CostDensity{\Prof}{2}$, we again find that \cref{eqn:ratioRightA}
holds.

Case~2: $8\Cost{\Prof}/n\leq\Gap{\Prof}\leq 2\CostDensity{\Prof}{2}$.
Thus $\Soln_{2,2}^*=\Gap{\Prof}$. Since
$\Soln_{1,2}^*=8\Cost{\Prof}/n$, we have $\Dist{x}{\Soln_1^*}\leq
2\Dist{x}{\Soln_2^*}$. Thus
\begin{eqnarray*}
\Dist{x}{\Soln^*}
& = &
\min(\Dist{x}{\Soln_1^*},\Dist{x}{\Soln_2^*})\\
& \geq &
\min(\Dist{x}{\Soln_1^*},\Dist{x}{\Soln_1^*}/2)\\
& = & 
\Dist{x}{\Soln_1^*}/2\\
& \geq &
\Dist{x}{\Soln_1}/2\\
& \geq &
\Dist{x}{\Soln}/2,
\end{eqnarray*}
where the second inequality follows from \cref{eqn:ratioLeft}.

Recall that $\Weight{\Prof}{2}\geq s_2/2$. Thus the case condition
implies $\Soln_{2,2}^*\leq 4\Cost{\Prof}/s_2$.  Recall that
$|\SolnOne^*_{2,1}-\SolnOne_{2,1}|\leq 2\Cost{\Prof}/s_2$.  Using
Fact~\ref{fact:spHelper} with $p_2$ equal to $4\Cost{\Prof}/s_2$,
$q_2$ equal to $\Cost{\Prof}/s_2$, and $z$ equal to
$2\Cost{\Prof}/s_2$, we find that
\begin{equation}
\label{eqn:ratioRightB}
\frac{\Dist{x}{\Soln_2^*}}{\Dist{x}{\Soln_2}}\leq 6.
\end{equation}
Thus
\begin{eqnarray*}
\Dist{x}{\Soln^*}
& = &
\min(\Dist{x}{\Soln_1^*},\Dist{x}{\Soln_2^*})\\
& \leq &
\min(10\Dist{x}{\Soln_1},6\Dist{x}{\Soln_2})\\
& \leq &
10\min(\Dist{x}{\Soln_1},\Dist{x}{\Soln_2})\\
& = &
10\Dist{x}{\Soln},
\end{eqnarray*}
where the first inequality follows from \cref{eqn:ratioLeft,eqn:ratioRightB}.

Case~3: $\Gap{\Prof}\leq 8\Cost{\Prof}/n$.  Thus
$\Soln_{2,2}^*=8\Cost{\Prof}/n$ and since
$\Soln_{1,2}^*=8\Cost{\Prof}/n$, we have $\Dist{x}{\Soln_1^*}\leq
2\Dist{x}{\Soln_2^*}$. As in Case~2, we deduce that
$\Dist{x}{\Soln^*}\geq\Dist{x}{\Soln}/2$.

Recall that $|\SolnOne^*_{2,1}-\SolnOne_{2,1}|\leq 2\Cost{\Prof}/s_2$.
Given that $1\leq s_2\leq n/2$ and using Fact~\ref{fact:spHelper} with $p_2$ equal to $8\Cost{\Prof}/n$,
$q_2$ equal to $\Cost{\Prof}/s_2$, and $z$ equal to
$2\Cost{\Prof}/s_2$, we find that
\cref{eqn:ratioRightB} holds as in Case~2. The rest of the argument
proceeds as in Case~2.
\end{proof}

The remainder of this section is devoted to establishing that
mechanism $\MechLipTwo$ is constant-Lipschitz.  The plan is to argue
Lipschitz-type bounds for the functions $\Fac{\Prof}{\ell}$ and
$\Backoff{\Prof}{\ell}$ that define the four components of the
solution produced by $\MechLipTwo$. These functions are defined in
terms of a number of simpler auxiliary functions; the latter functions
are the focus of much of our analysis in Appendix~\ref{sec:lip}.
Since the functions that we are studying all map a given $\Prof$ in
$\Reals^n$ to $\Reals$, a natural approach to deriving Lipschitz-type
bounds is to study the magnitude of the gradient. Instead of directly
reasoning about the gradient, we find it convenient to fix all
components except component $k$ of a given profile $\Prof$, and then
study the magnitude of the derivative as $\Prof_k$ varies over
$\IntervalA^*=\Interval{\Prof}{k}$.

A technical obstacle that we need to overcome is that most of the
functions we study are not differentiable over all of $\IntervalA^*$;
rather, they are piecewise differentiable.  In \cref{sec:knots}, we
break $\IntervalA^*$ into a finite set of pieces, each of which is a
closed interval, such that all of the functions of interest are
differentiable over the interior of each piece. In \cref{sec:cell}, we
use calculus to bound the magnitude of the derivative of each of these
functions over the interior of each piece. Combining these bounds with
some basic results presented in \cref{sec:lipFacts}, we obtain our
main technical lemma, \cref{lem:lipMainCell}, which implies that the
functions $\Fac{\Prof}{\ell}$ and $\Backoff{\Prof}{\ell}$ are
constant-Lipschitz over each piece.

The main result of \cref{sec:lip}, \cref{lem:lipMainHelper}, follows
easily from \cref{lem:lipMainCell} along with \cref{fact:lipCover}
from \cref{sec:lipFacts}; this lemma states that the functions
$\Fac{\Prof}{\ell}$ and $\Backoff{\Prof}{\ell}$ are constant-Lipschitz
over $\IntervalA^*$.  With \cref{lem:lipMainHelper} in hand, we easily
obtain \cref{lem:lipTech} below, which provides a Lipschitz-type bound
for mechanism~$\MechLipTwo$ when a profile $\Prof$ is changed to
another profile in $\Hood{\Prof}$. We then use \cref{lem:lipTech} to
obtain the desired constant-Lipschitz bound for $\MechLipTwo$ (see
\cref{lem:lipMain} below).

\begin{lemma}
\label{lem:lipTech}
There exists a positive constant $\kappa$ such that for any profile $\Prof$
and any profile $\Prof^*$ in $\Hood{\Prof}$, we have
\[
\|\MechLipTwo(\Prof^*)-\MechLipTwo(\Prof)\|_1\leq \kappa\cdot\|\Prof^*-\Prof\|_1.
\]
\end{lemma}
\begin{proof}
Immediate from Lemma~\ref{lem:lipMainHelper}.
\end{proof}

\begin{lemma}
\label{lem:lipIndex}
Let $\Prof$ and $\Prof^*$ be distinct $n$-profiles. Then there exists
an $i$ in $[n]$ such that $\Prof_i\not=\Prof_i^*$ and $\Prof_i^*$
belongs to $\Interval{\Prof}{i}$.
\end{lemma}
\begin{proof}
If there exists a $j$ in $[n]$ such that $\Prof_j<\Prof_j^*$, then we
can take $i$ to be the maximum such $j$. Otherwise, we can take $i$ to
be the minimum $j$ in $[n]$ such that $\Prof_j>\Prof_j^*$.
\end{proof}

We finally state our lemma that mechanism $\MechLipTwo$ is constant-Lipschitz as follows.

\begin{lemma}
\label{lem:lipMain}
Let $\kappa$ denote the Lipschitz constant of Lemma~\ref{lem:lipTech}. Then
for any $n$-profiles $\Prof$ and $\Prof^*$, we have
\[
\|\MechLipTwo(\Prof^*)-\MechLipTwo(\Prof)\|_1\leq \kappa\cdot\|\Prof^*-\Prof\|_1.
\]
\end{lemma}
\begin{proof}
Let $\Prof$ and $\Prof^*$ be $n$-profiles and let $s$ denote $|\{i\in
[n]\mid\Prof_i\not=\Prof_i^*\}|$.  By applying
Lemma~\ref{lem:lipIndex} $s$ times, we find that there exists a
sequence of $n$-profiles
$\Prof^{(0)},\ldots,\Prof^{(s)}$ such that $\Prof^{(0)}=\Prof$, $\Prof^{(s)}=\Prof^*$, and $\Prof^{(i)}$
belongs to $\Hood{\Prof^{(i-1)}}$ for all $i$ in $[s]$. Thus
\begin{eqnarray*}
\|\MechLipTwo(\Prof^*)-\MechLipTwo(\Prof)\|_1
& \leq &
\sum_{i\in [s]}\|\MechLipTwo(\Prof^{(i)})-\MechLipTwo(\Prof^{(i-1)})\|_1\\
& \leq &
\kappa\sum_{i\in [s]}\|\Prof^{(i)}-\Prof^{(i-1)}\|_1\\
& = & 
\kappa\cdot\|\Prof^*-\Prof\|_1,
\end{eqnarray*}
where the second inequality follows from Lemma~\ref{lem:lipTech}.
\end{proof}

\begin{theorem}
    \label{Mech4_Thm}
    Mechanism $\MechLipTwo$ is constant-approximate, constant-strategyproof, and constant-Lipschitz.
\end{theorem}
\begin{proof}
    The theorem follows immediately from \cref{lem:costLipTwo}, \cref{thm:spLip}, and \cref{lem:lipMain}.
\end{proof}

\section{Concluding Remarks}

Our work shows that constant-strategyproof and constant-Lipschitz mechanisms can provide good truthfulness and stability guarantees without suffering significant loss of efficiency. For example, consider the stark difference between a $\Theta(n)$ approximation ratio of optimal social cost for strategyproof mechanisms and a constant approximation ratio of optimal social cost for constant-strategyproof mechanisms. Our positive results suggest that these properties are deserving of further study for the facility location problem and other problems in mechanism design.

There are many interesting directions in which our work can be generalized. Most pressingly, we are eager to see whether our methods can be generalized to facility location in other metric spaces. In particular, we conjecture that it is possible to design constant-strategy proof mechanisms for facility location with both preferences and facilities in $\mathbb{R}^d$. Beyond this, there are other natural directions for generalizing our results such as allowing for more than two facilities or considering different cost functions for individual agent or social costs \cite{Deligkas2024,Kanellopoulos2023}.

There are also numerous other properties which are of interest. For example, future work could integrate recent fairness results for facility location \cite{Gupta2023,Hossain2020,Micha2020,Zhou2023} with constant-strategyproof and constant-Lipschitz mechanisms. Additionally, our results implicitly reveal a tradeoff between the approximation ratio and the relaxation of strategyproofness, governed by the choice of vertical offset from the line. It would be interesting to more precisely characterize this tradeoff.

The integration of theoretical results with practical, real-world, considerations is another potentially fruitful direction for future research. Examples of work in this vein include \cite{Candogan2024} which studies the effective use on individual mobility pattern data for facility location and \cite{Ashlagi2018,Gonczarowski2023,Li2017} which study ``obvious'' strategyproofness. Ultimately, facility location is a theoretical problem with deep practical motivations. Any work that attempts to narrow the gap between theory and implementation could be valuable both for embedding better theoretical guarantees into practical implementations and for informing which problems and relaxations are most deserving of theoretical study.

%
%
%
\bibliographystyle{splncs04}
\bibliography{references}
\appendix

\section{Some Results for Establishing Lipschitz-Type Bounds}
\label{sec:lipFacts}

In this appendix we state and prove some useful results for deriving
Lipschitz-type bounds. For any nonempty closed interval $\IntervalA$ of the real line, we write $\Interior{\IntervalA}$ to denote the largest open interval
contained in $\IntervalA$.

\begin{fact}
	\label{fact:lipCover}
	Let $\IntervalA$ be a closed interval, let
	$u:\IntervalA\rightarrow\Reals$ be a function over $\IntervalA$, and let
	$\IntervalA_1,\ldots,\IntervalA_s$ be a finite collection of closed
	subintervals of $\IntervalA$ such that $\IntervalA=\cup_{i\in[s]}
\IntervalA_i$. Further assume that $u$ is $\LipConstant$-Lipschitz over
	$\IntervalA_i$ for all $i$ in $[s]$. Then $u$ is
	$\LipConstant$-Lipschitz over $\IntervalA$.
\end{fact}

\begin{fact}
	\label{fact:lipDiff}
	Let $\IntervalA$ be a closed interval and let the function
	$u:\IntervalA\rightarrow\Reals$ be continuous over $\IntervalA$ and
	differentiable over $\Interior{\IntervalA}$ with $|u'(x)|\leq\LipConstant$ for
	all $x$ in $\Interior{\IntervalA}$.  Then $u$ is $\LipConstant$-Lipschitz over
	$\IntervalA$.
\end{fact}

\begin{fact}
	\label{fact:lipMaxMin} Let $\IntervalA$ be a closed interval
	and let the functions $u:\IntervalA\rightarrow\Reals$ and
	$v:\IntervalA\rightarrow\Reals$ be $\LipConstant_1$-Lipschitz
	and $\LipConstant_2$-Lipschitz, respectively, over
	$\IntervalA$.  Then the functions $\max(u,v)$ and $\min(u,v)$
	are each $\max(\LipConstant_1,\LipConstant_2)$-Lipschitz over
	$\IntervalA$.
\end{fact}

\begin{lemma}
	\label{lem:lip}
	Let $\IntervalA$ be a closed interval and let the functions
	$u:\IntervalA\rightarrow\Reals$ and $v:\IntervalA\rightarrow\Reals$ be
	continuous over $\IntervalA$ and differentiable over
	$\Interior{\IntervalA}$. Further assume that (1) $|u'(x)|\leq \LipConstant_1$ for
	all $x$ in $\Interior{\IntervalA}$ and (2) $|v'(x)|\leq \LipConstant_2$ for all
	$x$ in $\Interior{\IntervalA}$ such that $u(x)\geq v(x)$.  Then
	$\min(u(x),v(x))$ is $\max(\LipConstant_1,\LipConstant_2)$-Lipschitz over $\IntervalA$.
\end{lemma}
\begin{proof}
	Let $a$ and $b$ be distinct elements of $\IntervalA$ with $a<b$.  Let
	$\LipConstant$ denote $\max(\LipConstant_1,\LipConstant_2)$ and let $t(x)$ denote the function
	$\min(u(x),v(x))$. 
    We need to prove that $|t(b)-t(a)|\leq\LipConstant(b-a)$.  
    Let $X$ denote $\{x\in [a,b]\mid u(x)\leq v(x)\}$. If $X$ is empty, it
	follows from Fact~\ref{fact:lipDiff} that $t(x)$ is $\LipConstant_2$-Lipschitz
	over $[a,b]$ and hence $|t(b)-t(a)|\leq \LipConstant_2(b-a)\leq \LipConstant(b-a)$.
	
	Now assume that $X$ is nonempty. Let $a^*$ (resp., $b^*$)
	denote the infimum (resp., supremum) of the set $X$. We have
\begin{eqnarray*}
|t(b)-t(a)|
& = &
|t(a^*)-t(a)|+|t(b^*)-t(a^*)|+|t(b)-t(b^*)|\\
& = &
|t(a^*)-t(a)|+|u(b^*)-u(a^*)|+|t(b)-t(b^*)|\\
& \leq &
\LipConstant_2(a^*-a)+\LipConstant_1(b^*-a^*)+\LipConstant_2(b-b^*)\\
& \leq &
\LipConstant(b-a),
\end{eqnarray*}
where the first inequality holds because Fact~\ref{fact:lipDiff} implies that $u(x)$ is
$\LipConstant_1$-Lipschitz over $[a^*,b^*]$ and $t(x)$ is
$\LipConstant_2$-Lipschitz over intervals $[a,a^*]$ and
$[b^*,b]$.
\end{proof}

\section{Towards Lipschitz Continuity of Mechanism~$\MechLipTwo$}
\label{sec:lip}

Throughout this appendix, let $\Prof^*$ be an $n$-profile, let $k$
belong to $[n]$, let $\IntervalA^*$ denote $\Interval{\Prof^*}{k}$,
and let $\Prof$ be the same as $\Prof^*$ except that we regard
$\Prof_k$ as a variable with domain $\IntervalA^*$. Thus, as $\Prof_k$
varies over $\IntervalA^*$, $\Prof$ varies over the set of
$n$-profiles such that $\Prof_i=\Prof^*_i$ for all $i$ in $[n]-k$.
The goal of this section is to establish that $\Fac{\Prof}{\ell}$ and
$\Backoff{\Prof}{\ell}$ (which are functions of the lone variable
$\Prof_k$ with domain $\IntervalA^*$) are constant-Lipschitz for all
$\ell$ in $\{1,2\}$; see Lemma~\ref{lem:lipMainHelper} below.
This claim holds trivially if
$\IntervalA^*=\{\Prof^*_k\}$, so for the remainder of this appendix, we
assume that $\IntervalA^*$ properly contains $\Set{\Prof_k^*}$.

As in \cref{sec:lipFacts}, for any nonempty closed interval $\IntervalA$ of the real line, we write $\Interior{\IntervalA}$ to denote the largest open interval
contained in $\IntervalA$. We define a \emph{knot} as a real number in $\Interior{\IntervalA^*}$.
For any
finite set of knots $\Knots=\{x_1,\ldots,x_s\}$ where
$x_1<\cdots<x_s$, we define $\Cells{\Knots}$ as the set of closed
intervals, which we call cells, defined as follows. If $\Knots$
is empty, then $\Cells{\Knots}$ contains only the closed interval
$\IntervalA^*$.  Otherwise, $\Cells{\Knots}$ contains the following
$s+1$ closed intervals: the interval consisting of all real numbers in
$\IntervalA^*$ less than or equal to $x_1$; the closed interval
$[x_i,x_{i+1}]$ for each $i$ in $[s-1]$; the interval consisting of
all real numbers in $\IntervalA^*$ greater than or equal to $x_s$.

The remainder of \cref{sec:lip} is organized as
follows. \cref{sec:knots} defines a set of knots $\Knots_6$ for
partitioning the interval $\IntervalA^*$ so that various functions are
well-behaved (e.g., differentiable) throughout the interior of each of
the resulting cells. \cref{sec:cell} establishes Lipschitz-type bounds for various functions over a single such cell.

\subsection{Identifying a Suitable Set of Knots}
\label{sec:knots}

For each $i$ in $[n-1]$, there is a unique largest open interval
$\IntervalA_i$ in $\IntervalA^*$ such that $i$ belongs to
$\Active{\Prof}$ for all $\Prof_k$ in $\IntervalA_i$. We define the
set of knots $\Knots_1$ to include any endpoint of such a nonempty
interval $\IntervalA_i$ that is not an endpoint of $\IntervalA^*$.

For any $i$ in $[n-1]$, we say that solution $\Cand{\Prof}{i}$ is
\emph{stationary} if $k$ does not belong
to $\{\lceil i/2\rceil,\lceil (n+i)/2\rceil\}$ and is
\emph{non-stationary} otherwise.  In this section, we do not
need to concern ourselves with solution $\Cand{\Prof}{n}$ because
profile $\Prof$ is nontrivial for all $\Prof_k$ in
$\Interior{\IntervalA^*}$.

It is easy to see that $P_1(\Knots_1)$ holds, where the predicate
$P_1(\Knots)$ is defined below for any set of knots $\Knots$.

$P_1(\Knots)$: For any cell $\Cell$ in $\Cells{\Knots}$, there is a
subset $S$ of $[n-1]$ such that $\Active{\Prof}=S$ for all
$\Prof_k$ in $\Interior{\Cell}$. Moreover, for any $i$ in $S$,
$\CostSoln{\Prof}{\Cand{\Prof}{i}}$ is an affine function of $\Prof_k$
over $\Cell$, where the slope belongs to $\{-1,1\}$ (resp.,
$\{-1,0\}$) if $\Cand{\Prof}{i}$ is stationary (resp.,
non-stationary).

Consider how $\Index{\Prof}$ can change as $\Prof_k$ increases across
a cell in $\Cells{\Knots_1}$. Since $P_1(\Knots_1)$ holds, each time
$\Index{\Prof}$ changes, the slope of the associated cost function
decreases. Since the slope belongs to $\{-1,0,1\}$, there are at most
two such transitions per cell.  We construct the set of knots
$\Knots_2$ by starting with $\Knots_1$ and adding at most two knots
per cell in $\Cells{\Knots_1}$ at such transitions.  It is easy to see
that $\wedge_{1\leq i\leq 2}P_i(\Knots_2)$ holds, where the predicate
$P_2(\Knots)$ is defined below for any set of knots $\Knots$.

$P_2(\Knots)$: For any cell $\Cell$ in $\Cells{\Knots}$,
$\Index{\Prof}$ is the same for all $\Prof_k$ in
$\Interior{\Cell}$. Moreover, the number of agents served by the left
(resp., right) facility under $\Canonical{\Prof}$ remains the same for
all $\Prof_k$ in $\Interior{\Cell}$.

For any finite set of knots $\Knots$ containing $\Knots_2$, a cell
$\Cell$ in $\Cells{\Knots}$ is said to be increasing-cost (resp.,
decreasing-cost, constant-cost) if $\CostPrime{\Prof}=1$ (resp, $-1$,
$0$) for all $\Prof_k$ in $\Interior{\Cell}$.

\begin{lemma}
\label{lem:helper}
Let $\Cell$ be a cell in $\Cells{\Knots_2}$. Then the function
$\Prof_k-\Helper{\Prof}{1}$ is nondecreasing over $\Cell$. 
\end{lemma}
\begin{proof}
Let $x^{(1)}$ and $x^{(2)}$ be distinct values in $\Cell$.  Assume
that $x^{(1)}<x^{(2)}$ and let $\eps$ denote $x^{(2)}-x^{(1)}$.  For
each $i$ in $\{1,2\}$, let $\Prof^{(i)}$ denote the profile $\Prof^*$
with component $\Prof^*_k$ replaced by $x^{(i)}$.  Let $x$ denote
$\Helper{\Prof^{(1)}}{1}+\eps$.  We need to prove that
$\Helper{\Prof^{(2)}}{1}\leq x$.

Let $s$ denote the maximum index in $[n]$ such that
$\Prof^{(1)}_s\leq\Helper{\Prof^{(1)}}{1}$.  (Such an index exists
since $\Prof^{(1)}_1\leq\Helper{\Prof^{(1)}}{1}$.) We have (1)
$\Prof^{(1)}_i\leq\Helper{\Prof^{(1)}}{1}$ for all $i$ in $[s]$, (2)
$\Prof^{(2)}_i=\Prof^{(1)}_i$
for all $i$
in $[s]-k$, (3) $\Prof^{(2)}_k=\Prof^{(1)}_k+\eps$. Let $z$ denote
$\sum_{i\in [n]}\max(0,x-\Prof^{(2)}_i)$. We need to prove that
$z\geq\Cost{\Prof^{(2)}}$.  Observe that
$z\geq\Cost{\Prof^{(1)}}+(s-\Indicator{k\leq s})\eps$.  We consider
two cases.

Case~1: $k>1$ or $s>1$. In this case,
$z\geq\Cost{\Prof^{(1)}}+\eps$. Since Lemma~\ref{lem:lipCost} implies $\Cost{\Prof^{(2)}}\leq\Cost{\Prof^{(1)}}+\eps$, we conclude that
$z\geq\Cost{\Prof^{(2)}}$, as required.

Case~2: $k=s=1$.  In this case,
$z\geq\Cost{\Prof^{(1)}}+\eps-\eps=\Cost{\Prof^{(1)}}$ and so it is
sufficient to prove that $\Cost{\Prof^{(1)}}\geq\Cost{\Prof^{(2)}}$.
Observe that $\Cost{\Prof^{(1)}}=\Helper{\Prof^{(1)}}{1}-x^{(1)}$ and
$\Index{\Prof}=1$ throughout cell $\Cell$.  Since $k=1$, solution
$\Cand{\Prof}{1}$ is non-stationary and hence cell $\Cell$ is either
decreasing-cost or constant-cost (in fact, cell $\Cell$ is
constant-cost, but we won't need to argue that here), ensuring that
$\Cost{\Prof^{(1)}}\geq\Cost{\Prof^{(2)}}$, as required.
\end{proof}

Lemma~\ref{lem:helper} implies that as $\Prof_k$ increases across a
given cell in $\Cells{\Knots_2}$, there is at most one value where a
transition occurs from $\Prof_k\leq\Helper{\Prof}{1}$ to
$\Prof_k>\Helper{\Prof}{1}$.  Let $\Knots_3$ denote the set of knots
obtained by starting with $\Knots_2$ and adding a knot at each such
transition value.  It is easy to see that $\wedge_{1\leq i\leq
  3}P_i(\Knots_3)$ holds, where the predicate $P_3(\Knots)$ is defined
below for any set of knots $\Knots$.

$P_3(\Knots)$: For any cell $\Cell$ in $\Cells{\Knots}$, either (1)
$\Prof_k\leq\Helper{\Prof}{1}$ for all $\Prof_k$ in $\Interior{\Cell}$
or (2) $\Prof_k>\Helper{\Prof}{1}$ for all $\Prof_k$ in
$\Interior{\Cell}$.

For any finite set of knots $\Knots$ containing $\Knots_3$, a cell
$\Cell$ in $\Cells{\Knots}$ is said to be \emph{left-sided} (resp.,
\emph{right-sided}) if $\Prof_k\leq\Helper{\Prof}{1}$ (resp.,
$\Prof_k>\Helper{\Prof}{1}$) for all $\Prof_k$ in $\Interior{\Cell}$.

For any finite set of knots $\Knots$ containing $\Knots_3$, a cell
$\Cell$ in $\Cells{\Knots}$ is said to be
\emph{$\HelperFn{1}$-increasing} (resp.,
\emph{$\HelperFn{1}$-decreasing}, \emph{$\HelperFn{1}$-constant}) if
$\HelperFn{1}$ is strictly increasing (resp., strictly decreasing,
constant) over $\Cell$.

\begin{lemma}
\label{lem:helperMono}
Let $\Knots$ be a finite set of knots that includes $\Knots_3$ and let
$\Cell$ be a cell in $\Cells{\Knots}$.  Then cell $\Cell$ is
$\HelperFn{1}$-increasing, $\HelperFn{1}$-decreasing, or
$\HelperFn{1}$-constant.
\end{lemma}
\begin{proof}
The following claims are straightforward to prove: if $\Cell$ is
either (1) increasing-cost or (2) constant-cost and left-sided, then
$\Cell$ is $\HelperFn{1}$-increasing; if $\Cell$ is decreasing-cost
and right-sided, then $\Cell$ is $\HelperFn{1}$-decreasing; otherwise,
$\Cell$ is $\HelperFn{1}$-constant.
\end{proof}

For any $\HelperFn{1}$-increasing or $\HelperFn{1}$-decreasing cell
$\Cell$ in $\Cells{\Knots_3}$ and any $i$ in $[n]-k$, there can be at
most one value of $\Prof_k$ in $\Interior{\Cell}$ where
$\Helper{\Prof}{1}=\Prof_i$. We construct the set of knots $\Knots_4$
by starting with $\Knots_3$ and adding at most one knot per cell in
$\Cells{\Knots_3}$ corresponding to such a value.

\begin{lemma}
\label{lem:helperTech}
Let $\Prof$ be an $n$-profile and let $\Cell$ be a cell in
$\Cells{\Knots_4}$. Then $\Helper{\Prof}{1}<\Prof_n$ for all $\Prof_k$
in $\Interior{\Cell}$.
\end{lemma}
\begin{proof}
It is easy to argue that $\Helper{\Prof}{1}\leq\Prof_n$ for all
$\Prof_k$ in $\Cell$.  Suppose that $\Helper{\Prof}{1}=\Prof_n$ for
some $\Prof_k$ in $\Interior{\Cell}$. The definition of $\Knots_4$
implies $k=n$. Since $P_3(\Knots_4)$ holds, we deduce that
$\Helper{\Prof}{1}=\Prof_n$ for all $\Prof_k$ in $\Interior{\Cell}$.
It follows that $\Cell$ is $\HelperFn{1}$-increasing and
$\HelperPrime{\Prof}{1}=1$ for all $\Prof_k$ in $\Interior{\Cell}$.
Hence $n=2$. But then $\Cost{\Prof}=0$ for all $\Prof_k$ in $\Cell$
and hence $\Helper{\Prof}{1}=\Prof_1$, a contradiction.
\end{proof}

Using Lemma~\ref{lem:helperTech}, it is easy to see that
$\wedge_{1\leq i\leq 4}P_i(\Knots_4)$ holds, where the predicate
$P_4(\Knots)$ is defined below for any set of knots $\Knots$.

$P_4(\Knots)$: For any $\HelperFn{1}$-increasing or
$\HelperFn{1}$-decreasing cell $\Cell$ in $\Cells{\Knots}$, we have:
(1) there is a unique index $i$ in $[n-1]$ such that
$\Prof_i\leq\Helper{\Prof}{1}<\Prof_{i+1}$ holds for all $\Prof_k$ in
$\Interior{\Cell}$; (2) $\Helper{\Prof}{1}$ is an affine function of
$\Prof_k$ throughout cell $\Cell$.  Moreover, the slope of this affine
function is determined by the associated index $i$ as follows: if
$\Cell$ is $\HelperFn{1}$-increasing and left-sided then it is either
increasing-cost and the slope is $2/i$ where $i\geq 2$, or it is
constant-cost and the slope is $1/i$; if $\Cell$ is
$\HelperFn{1}$-increasing and right-sided then it is increasing-cost
and the slope is $1/i$; if $\Cell$ is $\HelperFn{1}$-decreasing then
it is right-sided and decreasing-cost and the slope is $-1/i$.

Thus far we have focused on the function $\HelperFn{1}$. We can
introduce knots for $\HelperFn{2}$ by establishing lemmas for
$\HelperFn{2}$ analogous to
Lemmas~\ref{lem:helper}, \ref{lem:helperMono},
and~\ref{lem:helperTech} for $\HelperFn{1}$, and by introducing
predicates $P_3^*(\Knots)$ and $P_4^*(\Knots)$ for $\HelperFn{2}$
analogous to predicates $P_3(\Knots)$ and $P_4(\Knots)$ for
$\HelperFn{1}$.  We define the set of knots $\Knots_5$ as $\Knots_4$
plus the knots associated with $\HelperFn{2}$. It is easy to see that
$[\wedge_{1\leq i\leq 4}P_i(\Knots_5)]
\wedge P_3^*(\Knots_5)\wedge P_4^*(\Knots_5)$ holds.

For any cell $\Cell$ in $\Cells{\Knots_5}$, there is at most one value
of $\Prof_k$ in $\Interior{\Cell}$ such that
$\Helper{\Prof}{1}=\Mean{\Prof}$ (and hence also
$\Helper{\Prof}{2}=\Mean{\Prof}$); the reason is that
$\MeanPrime{\Prof}$, which is $1/n$, can never match the slope of
$\Helper{\Prof}{1}$ by $P_4(K_5)$.  We obtain the set of knots
$\Knots_6$ by starting with $\Knots_5$ and adding a knot at any such
value. These knots ensure that for each resulting cell $\Cell$ in
$\Cells{\Knots_6}$, either (1) $\Gap{\Prof}=0$ for all $\Prof_k$ in
$\Interior{\Cell}$ or (2) $\Gap{\Prof}>0$ for all $\Prof_k$ in
$\Interior{\Cell}$. We say that a cell satisfying (1) is
\emph{zero-gap} and a cell satisfying (2) is \emph{positive-gap}.
It is easy to see that
$[\wedge_{1\leq i\leq 4}P_i(\Knots_6)]
\wedge P_3^*(\Knots_6)\wedge P_4^*(\Knots_6)$ holds.

We now state the main lemma of \cref{sec:lip}. The proof makes
use of Lemma~\ref{lem:lipMainCell}, which is proven in \cref{sec:cell}.

\begin{lemma}
\label{lem:lipMainHelper}
Let $\ell$ belong to $\{1,2\}$.  Then $\Fac{\Prof}{\ell}$ and
$\Backoff{\Prof}{\ell}$ are each constant-Lipschitz over
$\IntervalA^*$.
\end{lemma}
\begin{proof}
Lemma~\ref{lem:lipMainCell} implies the existence of a constant
$\kappa$ such that $\Fac{\Prof}{\ell}$ and $\Backoff{\Prof}{\ell}$ are
$\kappa$-Lipschitz over all cells in $\Cells{\Knots_6}$. Since the
union of the cells in $\Cells{\Knots_6}$ is $\IntervalA^*$ and each
pair of successive cells $\Cells{\Knots_6}$ overlap at a knot in
$\Knots_6$, it is straightforward to prove that $\Fac{\Prof}{\ell}$
and $\Backoff{\Prof}{\ell}$ are each $\kappa$-Lipschitz over
$\IntervalA^*$ using \cref{fact:lipCover}.
\end{proof}

\subsection{Lipschitz Bounds for a Single Cell}
\label{sec:cell}

Let $\Cell$ be a cell in $\Cells{\Knots_6}$. Since $P_2(\Knots_6)$
holds, $\Index{\Prof}$ is the same for all $\Prof_k$ in
$\Interior{\Cell}$ and the number of agents served by the left (resp.,
right) facility under $\Canonical{\Prof}$ is the same for all
$\Prof_k$ in $\Interior{\Cell}$.  Throughout the remainder
of \cref{sec:cell}, let $s_1$ (resp., $s_2$) denote the number of
agents served by the left (resp., right) facility under
$\Canonical{\Prof}$.

It is easy to see that $\MeanPrime{\Prof}=1/n$ for all $\Prof_k$ in
$\Cell$ and that $\Mean{\Prof}$ is $\frac{1}{n}$-Lipschitz over
$\Cell$.  Since $P_1(\Knots_6)$ holds, we have
$|\CostPrime{\Prof}|\leq 1$ for all $\Prof_k$ in $\Interior{\Cell}$.


Using Lemmas~\ref{lem:nest} and~\ref{lem:helperMono} and the fact that
$P_4(\Knots_6)$ holds, we find that
$|\HelperPrime{\Prof}{1}|=O(1/s_1)$ for all $\Prof_k$ in
$\Interior{\Cell}$. An analogous argument establishes that
$|\HelperPrime{\Prof}{2}|=O(1/s_2)$ for all $\Prof_k$ in
$\Interior{\Cell}$.  It follows from Fact~\ref{fact:lipDiff} that
$\Helper{\Prof}{\ell}$ is $O(1/s_{\ell})$-Lipschitz over $\Cell$ for
all $\ell$ in $\{1,2\}$.

\begin{lemma}
\label{lem:lipZeroGap}
Let $\ell$ belong to $\{1,2\}$.  If $\Cell$ is zero-gap, then
$\Fac{\Prof}{\ell}$ and $\Backoff{\Prof}{\ell}$ are each
$O(1/n)$-Lipschitz over $\Cell$.
\end{lemma}
\begin{proof}
Since $\Cell$ is zero-gap, we have $\Fac{\Prof}{\ell}=\Mean{\Prof}$
over $\Cell$. Since $\MeanPrime{\Prof}=1/n$, we conclude from
Fact~\ref{fact:lipDiff} that $\Fac{\Prof}{\ell}$ is
$\frac{1}{n}$-Lipschitz over $\Cell$.

Since $\Cell$ is zero-gap, we have $\BackoffHelper{\Prof}{\ell}=0$ and
$\Backoff{\Prof}{\ell}=8\Cost{\Prof}/n$ for all $\Prof_k$ in $\Cell$.
It follows from Lemma~\ref{lem:lipCost} that $\Backoff{\Prof}{\ell}$
is $O(1/n)$-Lipschitz over $\Cell$.
\end{proof}

\begin{lemma}
\label{lem:lipPosGapFac}
Let $\ell$ belong to $\{1,2\}$ and assume that $\Cell$ is
positive-gap. Then $|\FacPrime{\Prof}{\ell}|=O(1/s_{\ell})$ for all
$\Prof_k$ in $\Interior{\Cell}$ and $\Fac{\Prof}{\ell}$ is
$O(1/s_{\ell})$-Lipschitz over $\Cell$.
\end{lemma}
\begin{proof}
Since $\Cell$ is positive-gap,
$\Fac{\Prof}{\ell}=\Helper{\Prof}{\ell}$ over $\Interior{\Cell}$.
\end{proof}

\begin{lemma}
\label{lem:lipPosGapGap}
Let $\ell$ belong to $\{1,2\}$ and assume that $\Cell$ is
positive-gap. Then $|\GapPrime{\Prof}|=O(1/\min(s_1,s_2))$ for all
$\Prof_k$ in $\Interior{\Cell}$.
\end{lemma}
\begin{proof}
Immediate from Lemma~\ref{lem:lipPosGapFac}.
\end{proof}

\begin{lemma}
\label{lem:lipPosGapBigBackoff}
Let $\ell$ belong to $\{1,2\}$. Assume that $\Cell$ is positive-gap
and that $s_{\ell}\geq n/2$. Then $\Backoff{\Prof}{\ell}$ is
$O(1/n)$-Lipschitz over $\Cell$.
\end{lemma}
\begin{proof}
Since $s_{\ell}\geq n/2$, Lemma~\ref{lem:weightLowerBound} implies
$\Weight{\Prof}{\ell}\geq n/4$ for all $\Prof_k$ in
$\Interior{\Cell}$. Thus $\BackoffHelper{\Prof}{\ell}\leq
8\Cost{\Prof}/n$ and hence $\Backoff{\Prof}{\ell}=8\Cost{\Prof}/n$ for
all $\Prof_k$ in $\Interior{\Cell}$.  It follows from Lemma~\ref{lem:lipCost}
that $\Backoff{\Prof}{\ell}$ is $O(1/n)$-Lipschitz over $\Cell$.
\end{proof}

Before stating the next lemma, we introduce a couple of useful
definitions.  Since $P_4(\Knots_6)$ holds, there is a subset $A$ of
$[n]$ such that $\{i\in[n]\mid\Prof_i<\Fac{\Prof}{1}\}=A$ for all
$\Prof_k$ in $\Interior{\Cell}$.  Likewise, since $P_4^*(\Knots_6)$
holds, there is a subset $B$ of $[n]$ such that
$\{i\in[n]\mid\Prof_i>\Fac{\Prof}{2}\}=B$ for all $\Prof_k$ in
$\Interior{\Cell}$. It follows that for all $i$ in $[n]\setminus
(A\cup B)$, we have $\Fac{\Prof}{1}\leq\Prof_i\leq\Fac{\Prof}{2}$ for
all $\Prof_k$ in $\Interior{\Cell}$.

\begin{lemma}
\label{lem:lipPosGapFrac}
Let $i$ belong to $[n]$ and let $\ell$ belong to $\{1,2\}$. Assume
that $\Cell$ is positive-gap and $s_{\ell}\leq n/2$.  If $i$ belongs
to $A\cup B$, then $\FracPrime{\Prof}{\ell}{i}=0$ for all $\Prof_k$ in
$\Interior{\Cell}$.  Otherwise,
\[
|\FracPrime{\Prof}{\ell}{i}| \leq
\frac{\Indicator{i=k}+O(1/n)+\Frac{\Prof}{\ell}{i}\cdot
O(1/s_{\ell})}{\Gap{\Prof}}
\]
for all $\Prof_k$ in $\Interior{\Cell}$.
\end{lemma}
\begin{proof}
If $i$ belongs to $A\cup B$, it is straightforward to verify that
$\FracPrime{\Prof}{\ell}{i}=0$ for all $\Prof_k$ in
$\Interior{\Cell}$. For the rest of the proof, assume that $i$ belongs
to $[n]\setminus (A\cup B)$. As discussed above, this means that
$\Fac{\Prof}{1}\leq\Prof_i\leq\Fac{\Prof}{2}$ for all $\Prof_k$ in
$\Interior{\Cell}$.

We address the case $\ell=2$; the case $\ell=1$ is symmetric. Since $\ell=2$, we have $s_1\geq n/2$. Hence
Lemma~\ref{lem:lipPosGapFac} implies
$|\FacPrime{\Prof}{1}|=O(1/n)$. Using Lemma~\ref{lem:lipPosGapGap}, we
find that
\begin{eqnarray*}
|\FracPrime{\Prof}{2}{i}|
 & \leq & 
\frac{\Indicator{i=k}+O(1/n)}{\Gap{\Prof}}
+\frac{\Prof_i-\Fac{\Prof}{1}}{\Gap{\Prof}}\cdot\frac{O(1/s_2)}{\Gap{\Prof}}\\
 & \leq &
\frac{\Indicator{i=k}+O(1/n)+\Frac{\Prof}{2}{i}\cdot O(1/s_2)}{\Gap{\Prof}}
\end{eqnarray*}
for all $\Prof_k$ in $\Interior{\Cell}$, as required.
\end{proof}

\begin{lemma}
\label{lem:lipPosGapSmallWeight}
Let $\ell$ belong to $\{1,2\}$. Assume that $\Cell$ is positive-gap
and that $s_{\ell}\leq n/2$.  Then $|\WeightPrime{\Prof}{\ell}|$ is
$O(\frac{\Weight{\Prof}{\ell}}{s_{\ell}\Gap{\Prof}})$
for all $\Prof_k$ in $\Interior{\Cell}$.
\end{lemma}
\begin{proof}
Summing the bound of Lemma~\ref{lem:lipPosGapFrac} over all
$i$ in $[n]$, we obtain
\begin{eqnarray*}
|\WeightPrime{\Prof}{\ell}|
& = &
O(1/\Gap{\Prof})+O\left(\frac{1}{s_{\ell}\Gap{\Prof}}\right)\cdot
\sum_{i\in [n]}\Frac{\Prof}{\ell}{i}\\
& = & 
O(1/\Gap{\Prof})+O\left(\frac{\Weight{\Prof}{\ell}}{s_{\ell}\Gap{\Prof}}\right)\\
& = & 
O\left(\frac{\Weight{\Prof}{\ell}}{s_{\ell}\Gap{\Prof}}\right)
\end{eqnarray*}
for all $\Prof_k$ in $\Interior{\Cell}$, where the last equation
follows from Lemma~\ref{lem:weightLowerBound}.
\end{proof}

\begin{lemma}
\label{lem:lipPosGapSmallCostDensity}
Let $\ell$ belong to $\{1,2\}$. Assume that $\Cell$ is positive-gap and that
$s_{\ell}\leq n/2$.  Then $|\CostDensityPrime{\Prof}{\ell}|$ is
$O(1+\CostDensity{\Prof}{\ell}/\Gap{\Prof})/s_{\ell}$
for all $\Prof_k$ in $\Interior{\Cell}$.
\end{lemma}
\begin{proof}
Since $\Cost{\Prof}$ and $\Weight{\Prof}{\ell}$ are each
differentiable for all $\Prof_k$ in $\Interior{\Cell}$,
$\CostDensity{\Prof}{\ell}$ is also differentiable for all $\Prof_k$
in $\Interior{\Cell}$. Moreover, for all $\Prof_k$ in $\Interior{\Cell}$,
we have
\begin{eqnarray*}
|\CostDensityPrime{\Prof}{\ell}|
& \leq &
\frac{|\CostPrime{\Prof}|}{\Weight{\Prof}{\ell}}
+\frac{\Cost{\Prof}\WeightPrime{\Prof}{\ell}}{\Weight{\Prof}{\ell}^2}\\
& \leq &
\frac{1}{\Weight{\Prof}{\ell}}
+O\left(\frac{\CostDensity{\Prof}{\ell}}{s_{\ell}\Gap{\Prof}}\right)\\
& = &
O(1+\CostDensity{\Prof}{\ell}/\Gap{\Prof}/s_{\ell},
\end{eqnarray*}
where the second inequality follows from $|\CostPrime{\Prof}|\leq 1$
and Lemma~\ref{lem:lipPosGapSmallWeight}, and the last step follows from
Lemma~\ref{lem:weightLowerBound}.
\end{proof}

\begin{lemma}
\label{lem:lipPosGapSmallBackoffHelper}
Let $\ell$ belong to $\{1,2\}$. Assume that $\Cell$ is positive-gap
and that $s_{\ell}\leq n/2$.  Then $\BackoffHelper{\Prof}{\ell}$ is
$O(1/s_{\ell})$-Lipschitz over $\Cell$.
\end{lemma}
\begin{proof}
By Lemma~\ref{lem:lipPosGapGap}, we know that $|\GapPrime{\Prof}|$ is
$O(1/s_{\ell})$ for all $\Prof_k$ in $\Interior{\Cell}$. By
Lemma~\ref{lem:lipPosGapSmallCostDensity}, we know that
$|\CostDensityPrime{\Prof}{\ell}|$ is $O(1/s_{\ell})$ for all
$\Prof_k$ in $\Interior{\Cell}$ such that $\Gap{\Prof}\geq
2\CostDensity{\Prof}{\ell}$.  Applying Lemma~\ref{lem:lip} with
$\Gap{\Prof}$ and $2\CostDensity{\Prof}{\ell}$ playing the roles of
the functions $u$ and $v$, respectively, we find that the claim of the
lemma holds.
\end{proof}

\begin{lemma}
\label{lem:lipPosGapSmallBackoff}
Let $\ell$ belong to $\{1,2\}$. Assume that $\Cell$ is positive-gap
and that $s_{\ell}\leq n/2$.  Then $\Backoff{\Prof}{\ell}$ is
$O(1/s_{\ell})$-Lipschitz over $\Cell$.
\end{lemma}
\begin{proof}

It follows from Lemma~\ref{lem:lipCost} that $\Cost{\Prof}/n$ is
$O(1/n)$-Lipschitz over $\Cell$. Thus the claim of the lemma follows
from Lemma~\ref{lem:lipPosGapSmallBackoffHelper} and
Fact~\ref{fact:lipMaxMin}.
\end{proof}

\begin{lemma}
\label{lem:lipMainCell}
Let $\ell$ belong to $\{1,2\}$.  Then $\Fac{\Prof}{\ell}$ and
$\Backoff{\Prof}{\ell}$ are each $O(1/s_{\ell})$-Lipschitz over
$\Cell$.
\end{lemma}
\begin{proof}
Immediate from Lemmas~\ref{lem:lipZeroGap}, \ref{lem:lipPosGapFac},
\ref{lem:lipPosGapBigBackoff}, and~\ref{lem:lipPosGapSmallBackoff}.
\end{proof}

\end{document}